\documentclass[conference]{IEEEtran}

\IEEEoverridecommandlockouts

\usepackage{amsthm}
\usepackage{url} 
\usepackage{epstopdf}

\usepackage[linesnumbered,ruled,lined,boxed,noend]{algorithm2e}
\newcommand{\myparagraph}[1]{\vspace{0.5\baselineskip}\noindent{\textbf{#1.}}~}
\usepackage{xspace}
\newcommand{\Baidu}{{\footnotesize{\textsf{Baidu-dataset}}}\xspace}
\newcommand{\BTAA}{{\footnotesize{\textsf{BTAA-dataset}}}\xspace}
\newcommand{\NYU}{{\footnotesize{\textsf{NYU-dataset}}}\xspace}
\newcommand{\Transit}{{\footnotesize{\textsf{Transit-dataset}}}\xspace}
\newcommand{\UMN}{{\footnotesize{\textsf{UMN-dataset}}}\xspace}

\newcommand{\quadtree}{{\small{\textsf{QuadTree}}}\xspace}
\newcommand{\Rtree}{{\small{\textsf{Rtree}}}\xspace}

\newcommand{\unifiedIndex}{{\small{\textsf{DITS}}}\xspace}
\newcommand{\localIndex}{{\small{\textsf{\mbox{DITS-L}}}}\xspace}
\newcommand{\globalIndex}{{\small{\textsf{\mbox{DITS-G}}}}\xspace}
\newcommand{\overlapSearch}{{\small{\textsf{OverlapSearch}}}\xspace}
\newcommand{\coverageSearch}{{\small{\textsf{CoverageSearch}}}\xspace}
\newcommand{\STS}{{\small{\textsf{STS3}}}\xspace}
\newcommand{\Josie}{{\small{\textsf{Josie}}}\xspace}

\newcommand{\SG}{{\small{\textsf{SG}}}\xspace}

\newcommand{\MIQ}{{{\small\textbf{\textrm{OJSP}}}}\xspace}
\newcommand{\MCQC}{{{\small\textbf{\textrm{CJSP}}}}\xspace}
\newcommand{\MCP}{{{\small\textrm{MCP}}}\xspace}
\SetKwComment{Comment}{$\triangleright$\ }{}

\newtheorem{example1}{Example}
\newtheorem{theorem}{Theorem}
\newtheorem{definition1}{Definition}

\newtheorem{lemma1}{Lemma}

\usepackage{subfigure}
\makeatletter
\renewcommand{\@thesubfigure}{\hskip\subfiglabelskip}
\makeatother
\usepackage{epsfig}
\usepackage{epstopdf}
\usepackage[colorlinks=true,linkcolor=ACMPurple,urlcolor=ACMPurple,citecolor=ACMPurple]{hyperref} 

\usepackage{enumitem}
\usepackage{titlesec}

\usepackage{color,soul}
\RequirePackage[prologue]{xcolor}
\definecolor[named]{ACMPurple}{cmyk}{0.55,1,0,0.15}
\definecolor{Maroon}{cmyk}{0, 0.87, 0.68, 0.32}
\definecolor{Brown}{rgb}{0.59, 0.29, 0.0}
\definecolor{Green}{rgb}{0,1,0}
\definecolor{NavyBlue}{rgb}{0.0, 0.0, 0.5}

\setul{0.5ex}{0.3ex}
\definecolor{Green}{rgb}{0,1,0}
\setulcolor{Green}

\newcommand\blfootnote[1]{%
  \begingroup
  \renewcommand\thefootnote{}\footnote{#1}%
  \addtocounter{footnote}{-1}%
  \endgroup
}

\definecolor{issuecolor}{RGB}{0,166,81}

\newcounter{cN}
\usepackage{etoolbox}
\usepackage{marginnote}
\makeatletter
\let\oldmarginnote\marginnote
\renewcommand*{\marginnote}[1]{%
	\begingroup%
	\ifodd\value{page}
	\if@firstcolumn\normalmarginpar\fi
	\else
	\if@firstcolumn\else\normalmarginpar\fi
	\fi
\textbf{}	\oldmarginnote{\textcolor{Brown}{#1}}%
		\oldmarginnote{}
	\endgroup%
}
\makeatother
\usepackage{soul}
\usepackage{etoolbox}

\usepackage{balance}
\usepackage{tablefootnote}
\usepackage{booktabs} 
\usepackage{tikz}  

\usepackage{hyperref}
\hypersetup{pdfstartpage=1}  

\usepackage{cite}
\usepackage{amsmath,amssymb,amsfonts}
\usepackage{algorithmic}
\usepackage{graphicx}
\usepackage{textcomp}
\usepackage{xcolor}
\def\BibTeX{{\rm B\kern-.05em{\sc i\kern-.025em b}\kern-.08em
    T\kern-.1667em\lower.7ex\hbox{E}\kern-.125emX}}

\begin{document}

\title{Joinable Search over Multi-source Spatial Datasets: Overlap, Coverage, and Efficiency
}

\author{\IEEEauthorblockN{Wenzhe Yang\IEEEauthorrefmark{2}, Sheng Wang\IEEEauthorrefmark{2}\IEEEauthorrefmark{1}, Zhiyu Chen\IEEEauthorrefmark{3}, Yuan Sun\IEEEauthorrefmark{4}, and Zhiyong Peng\IEEEauthorrefmark{2}\IEEEauthorrefmark{5}\IEEEauthorrefmark{1}}
\IEEEauthorblockA{\IEEEauthorrefmark{2}School of Computer Science, Wuhan University \;  
\IEEEauthorrefmark{3}Amazon.com, Inc. \\ 
\IEEEauthorrefmark{4}La Trobe Business School, La Trobe University \; \IEEEauthorrefmark{5}Big Data Institute,Wuhan University \\ }
\IEEEauthorblockA{[wenzheyang, swangcs, peng]@whu.edu.cn, zhiyuche@amazon.com, yuan.sun@latrobe.edu.au}
}

\maketitle

\begin{abstract}
The search for joinable data is pivotal for numerous applications, such as data integration, data augmentation, and data analysis. Although there have been many successful joinable search studies for table discovery, the study of finding joinable spatial datasets for a given query from multiple spatial data sources has not been well considered.
This paper studies two cases of joinable search problems from multiple spatial data sources. In addition to the overlap joinable search problem (\MIQ), we also propose a novel coverage joinable search problem (\MCQC) that has not been considered before, motivated by many real-world applications in the field of spatial search. 
To support two cases of joinable search over multiple spatial data sources seamlessly, we propose a multi-source spatial dataset search framework. Firstly, we design a 
\underline{\textbf{DI}}stributed 
\underline{\textbf{T}}ree-based \underline{\textbf{S}}patial index structure called \unifiedIndex, which is used not only to design acceleration strategies to speed up joinable searches, but also to support efficient communication between multiple data sources. Additionally, we prove that the \MCQC is NP-hard and design a greedy approximate algorithm to solve the problem. 
We evaluate the efficiency of our search framework on five real-world data sources, and the experimental results show that our framework can significantly reduce running time and communication costs compared with baselines.

\end{abstract}


\section{Introduction}
\label{sec:intro}

Spatial data is a critical component of real-world data~\cite{Yang2022,FangMS2023,ChowdhuryMS22,LiuSC21}. Efficient spatial queries have become increasingly important in spatial data studies~\cite{Wu2021,Vu2020,Zacharatou2019,kalamatianos2021}. For example, the spatial join focuses on finding all matched pairs of records that satisfy some specified join predicate for two given spatial datasets~\cite{Qiao2020,Zacharatou2019,ChenCCT15}. Unfortunately, the traditional join operations require users to specify the datasets to be joined, posing challenges for junior users who are unfamiliar with the data source's data. Thus, the joinable table search ~\cite{zhu2019,Dong2021,ZhangYi2020,DengDong2017} has been proposed to identify the tables that can be joined with the query table, which is a key procedure in subsequent data analysis~\cite{Dong2023,DengCCYCYSWLCJZJZWYWT24}. However, the joinable search over spatial datasets has not been well-studied.

\blfootnote{$^{\ast}$ Sheng Wang and Zhiyong Peng are the corresponding authors.}
\blfootnote{\IEEEauthorrefmark{3} Work does not relate to the position at Amazon.}
Especially nowadays, massive spatial data often come from different data sources and are located in different regions. Many open-source data platforms, such as OpenGeoMetadata~\cite{OpenGeoMetadata}, GeoBlacklight~\cite{GeoBlacklight}, OpenGeoHub~\cite{OpenGeoHub}, etc., integrate spatial data from multiple open data sources. Also, individual companies or organizations are independent in managing datasets and have a wealth of spatial data stored in their own data sources. It is crucial to break the barriers between different data sources to facilitate joinable search over different spatial data sources. Motivated by this, we focus on the joinable search problem over multiple spatial data sources. In what follows, we will use a municipal planning example to introduce two types of joinable searches.

\begin{figure}
\setlength{\abovecaptionskip}{-0.1 cm} 
\setlength{\belowcaptionskip}{0 cm} 
\vspace{-0.5em}
\centering

\includegraphics[width=9.0cm]{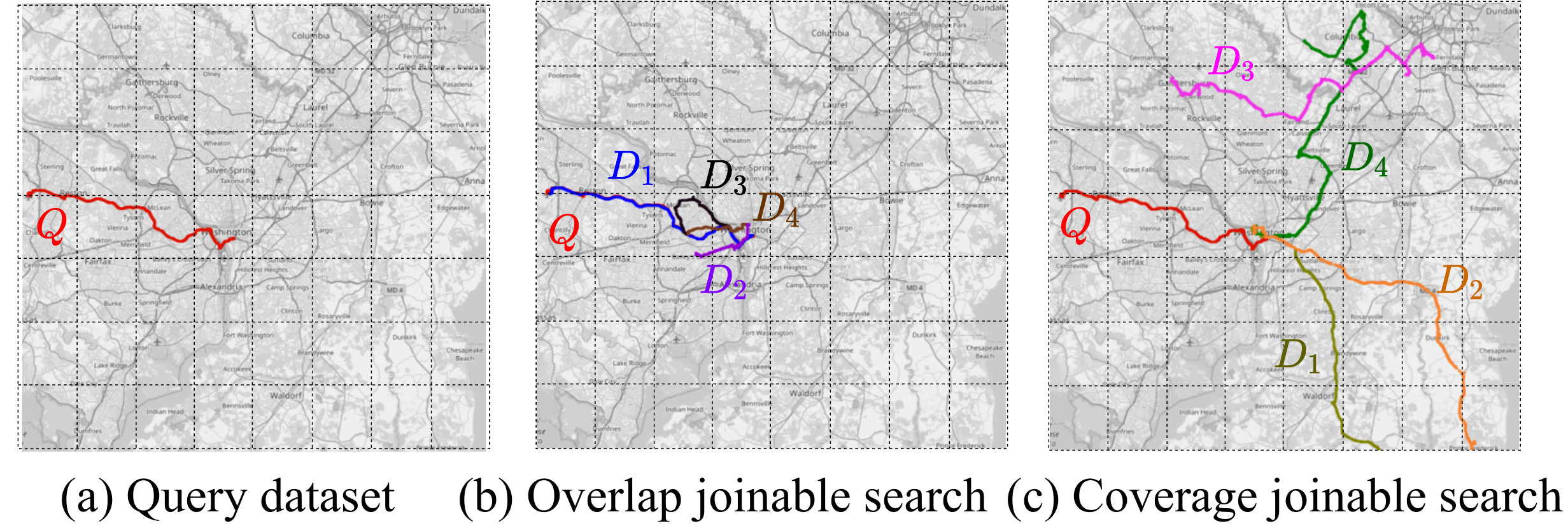}
\caption{An example of performing overlap and coverage joinable search on real multi-source transit datasets from Washington, D.C. and Maryland.
} \label{fig:intro}

\end{figure}

\begin{example1}
\label{example1}
In municipal planning, planners need to find datasets for analyzing traffic patterns and building a sound transportation system. There are two main tasks: (1) as shown in Fig.~\ref{fig:intro}(a), a user inputs a query dataset $Q$ consisting of a set of spatial points in Washington, D.C. 
The user specifies the longitude and latitude as the join columns and wants to find 4 datasets that have the maximum spatial overlap with $Q$. The results are shown in Fig.~\ref{fig:intro}(b). 
On this basis, we can utilize the found joinable datasets and $Q$ for subsequent data analysis, such as trajectory near-duplicate detection~\cite{XiaoWLYW11} and traffic congestion prediction~\cite{ZhengZYSZ14,ZhouTWN13}.
We define this type of search as an \textbf{overlap joinable search}; (2) the user also wants to find 4 datasets that have maximum spatial coverage after connecting to with $Q$, which can help construct transfer routes and meet more travel needs~\cite{Ali2018,He2022}. 
It is inappropriate for users to transfer for too long when switching transportation. Thus, a natural constraint is to ensure connectivity between different routes. As shown in Fig.~\ref{fig:intro}(c), the user can find 4 datasets of routes in Maryland that are connected to the query and cover larger regions.
We define this type of search operation as a \textbf{coverage joinable search}.
\end{example1}

\noindent\textbf{What is Overlap Joinable Search?} The overlap joinable search is similar to the overlap set similarity search problem~\cite{zhu2019,ZhangYi2020,ZhuEK2016}, whose goal is to find $k$ sets that have the maximum intersection size with the query set~\cite{Qiao2020,ZhangYG2017, Castelo2021}. Unfortunately, the spatial dataset is a set of points with longitude and latitude. The mere comparison of the numerical values of two coordinates is insufficient in determining whether the two points overlap or match, as geospatial data typically exhibits a certain degree of precision error in the decimal places. 

To facilitate the calculation of the intersection number of the spatial datasets, we partition the space into uniform cells (as shown in the dashed cells of Fig.~\ref{fig:intro}(a)). The size of the set intersection between two spatial datasets is the number of overlapping cells. Fig.~\ref{fig:intro}(b) shows 4 routes in Washington, D.C. with the maximum overlapping with $Q$. We can see that the query $Q$ and the route $D_1$ are highly overlapping.





\noindent\textbf{What is Coverage Joinable Search?}
The goal of coverage joinable search is to find $k$ datasets that have the maximum coverage area with the query dataset.
When performing the coverage joinable search, as points have no physical dimensions (length and width), it is difficult to measure the coverage area of the spatial dataset. Thus, we can define the number of cells containing these spatial points as the coverage of the spatial dataset. Additionally, spatial connectivity~\cite{WangTT2022,wang2021public,ChenChen2018} is a key constraint in the coverage joinable search, ensuring the relevance between the results with the query. Motivated by the above example, a general definition of spatial connectivity is necessary. 


The spatial connectivity we define requires that any result dataset be directly or indirectly connected with the query, where a dataset is directly connected to a query if two spatial datasets satisfy the given criteria (e.g., there exists at least one overlapped cell), and a dataset is indirectly connected to a query if one or more intermediate datasets are directly connected to the query as a medium. For example, in Fig.~\ref{fig:intro}(c), for the query $Q$ in Washington, D.C., we can find 4 routes in Maryland with the maximum coverage and satisfy the spatial connectivity, where $Q$ and $D_4$ are directly connected since they have an overlapped cell, while $Q$ and $D_3$ are indirectly connected since they can establish a connection through the intermediary datasets $D_4$. 

\myparagraph{Two Problems to Solve}
In this paper, we take the first attempt to define two types of joinable search problems over spatial datasets to enrich a given query dataset in both depth (overlapping) and width (coverage) directions. The former one can be formulated as the \textit{overlap joinable search problem} (\textbf{\MIQ}), and the latter one as the \textit{coverage joinable search problem} (\textbf{\MCQC}). In addition, considering that the data is widely distributed on multiple data sources, it is not enough to only study the local joinable search. Thus, we study the multi-source dataset search that supports overlap and coverage joinable search operations.

\myparagraph{Challenges}
To solve the above problems, we face multiple key challenges: (1) a na\"{\i}ve way to support overlap and coverage joinable searches is to build individual indexes for each type of joinable search. However, this can be costly in terms of both storage space and index construction time. Hence, it is preferable to have an efficient index that supports both types of joinable searches; (2) data is often stored in multiple independent and autonomous data sources. Breaking the barriers between different spatial data sources to facilitate the joinable search is also a core challenge; (3) the scale of datasets in each data source is very large, so simply designing the index is not enough. Moreover, we also prove that \MCQC is NP-hard. Thus, it is urgent to design efficient pruning strategies combined with the index to accelerate the joinable search process over multiple data sources. 

\myparagraph{Contributions} To address the challenges above, we mainly make the following contributions:




\begin{itemize}
[leftmargin=*]
\item To support two types of joinable searches over multiple spatial data sources, we first formulate them into two problems called \MIQ and \MCQC, and prove the NP-hardness of \MCQC (see Section~\ref{sec:defnitions}). Then, we develop an efficient search framework to support \MIQ and \MCQC (see Section~\ref{sec:framework}).
\item We design a \underline{\textbf{DI}}stributed \underline{\textbf{T}}ree-based \underline{\textbf{S}}patial index structure \unifiedIndex, composed of local indices and a global index, which not only accelerates \MIQ and \MCQC locally, but also supports efficient communication between multiple data sources (see Section~\ref{sec:offline}).

\item Based on \unifiedIndex, we first design query distribution strategies to reduce communication costs. Moreover, we propose effective lower and upper bounds for \MIQ and \MCQC, and design an efficient exact algorithm and a greedy approximation algorithm to solve them, respectively (see Section~\ref{sec:static}).

\item We conduct extensive experiments on five real-world spatial data sources to evaluate the index construction, search performance, and communication cost of the multi-source joinable search framework. The experimental results show that our proposed framework is highly effective and efficient in solving \MIQ and \MCQC compared with the baselines (see Section~\ref{sec:exp}).
\end{itemize}

\vspace{-0.2cm}
\section{Related Work}
\label{sec:relatedWork}
There are three kinds of work closely related to the joinable search over spatial datasets, including spatial join operation~\cite{Jacox2007,QiBM13,BourosM2019,HjaltasonS98,NobariTHKBA13}, overlap set similarity search~\cite{zhu2019,ZhuEK2016,Peng2016,YangZhong2020}, and maximum coverage problem (\MCP)~\cite{Hochbaum1998,Vazirani2001,VandinUR2011,CohenK08}, wherein our coverage joinable search can be viewed as a variant of \MCP. In what follows, we will conduct a literature review in these three parts.



\myparagraph{Spatial Join} 
The spatial join operation is similar to the join operation in relational databases~\cite{HjaltasonS98,Zhu2005,Qiao2020,QiBM20,BourosM2019}, which is defined on two given sets of spatial objects (e.g., point, line, polygon), and finds all matched pairs of objects by their spatial relationships (most often intersection). For example, for two given spatial point sets, the spatial intersection join~\cite{BrinkhoffKS93,BourosM2019} aims to find all pairs of intersection points between the two datasets. In contrast, spatial joinable search identifies $k$ datasets from the data source that can be joined with the query dataset.

As the amount of data increased, the spatial join processing for distributed systems has attracted much attention~\cite{TsitsigkosBMT19}. For example, in the Hadoop-GIS system~\cite{Aji2013}, spatial joins are computed by dividing the space into a grid and the objects in each cell are stored locally at nodes in the HDFS. In the Spatial-Hadoop system~\cite{Eldawy2015}, a global index is stored at a master node, then the master node partitions all data into chunks and distributes the data partition to slave nodes. Moreover, implementations of the spatial join using Spark~\cite{Xie2016,You2015,Yu2015}, where the data is shared in the memories of all nodes. However, the existing frameworks only support the traditional spatial join operation, which cannot be used to accelerate and solve the overlap and coverage joinable search proposed in this paper simultaneously.

\myparagraph{Overlap Set Similarity Search}
The goal of overlapping set similarity search is to find $k$ sets that have the maximum set intersection size with the query~\cite{zhu2019,ZhangYi2020,ZhuEK2016,XiaoC2008,Shrivastava2015}. 
Currently, the existing set similarity search studies are mainly divided into exact search~\cite{zhu2019,Peng2016,bayardo2007scaling, XiaoC2008} and approximate search~\cite{ZhuEK2016,FernandezMNM19,Dong2023}.

For example, fuzzy join~\cite{ChenWNC2019,ChaudhuriGGM03, WangLF14,DengKMS17,Zeakis0SPK22} is a powerful operator used in set matching that allows records to match approximately~\cite{YuLDF2016}. 
An early solution is FastJoin~\cite{WangLF14}, which generates the signature for each set and finds similar pairs of records. SilkMoth~\cite{DengKMS17} optimizes this signature scheme, reducing the number of used tokens to generate fewer candidates. Zeakis et al.~\cite{Zeakis0SPK22} proposed a token-based novel filtering technique to improve efficiency. For spatial datasets, the coordinates can form a multi-column composite key theoretically~\cite{ZhangHOPS10,SantosBCMF21} to perform the approximate matching. Nevertheless, the existing fuzzy join studies~\cite{ChenWNC2019,ChaudhuriGGM03, WangLF14,DengKMS17,Zeakis0SPK22} mainly focus on finding approximate matching pairs in two given sets. There is still a lack of efficient indexes to accelerate the computation of fuzzy matching between a query and plenty of candidate datasets.

In this paper, we focus on the exact solution. 
For example, Peng et al.~\cite{Peng2016} transformed time series datasets into a set of cells and measured the similarity based on the Jaccard index. However, it requires scanning all datasets and estimating the number of set intersections, where pairwise comparisons are time-consuming. 
Zhu et al.~\cite{zhu2019} proposed an exact algorithm called Josie to accelerate the set intersection computation. However, Josie inherits the limitation of the prefix filter, whose performance highly depends on the data distribution and yields a worst-case time complexity. In addition, unlike table data, spatial datasets are widely distributed in space. Thus, we need to design an effective index and search algorithm for the characteristics of spatial datasets.

\myparagraph{Maximum Coverage Problem}
The maximum coverage problem (\MCP) is a classical problem in combinatorics~\cite{Hochbaum1998}. Given a universe of elements $\mathcal{U} = \{e_1, e_2,\dots, e_n\}$ and a collection of sets $\mathcal{S} = \{S_1, S_2,\dots, S_m\}$, where each $S_i$ is a subset of $\mathcal{U}$, the objective of \MCP is to select a fixed number ($k$) of sets, such that the total number of elements covered by the selected sets is maximized. 
Over the past two decades, \MCP and its variants have been extensively studied, such as the budgeted \MCP~\cite{KhullerMN99}, the weighted \MCP~\cite{Vazirani2001}, the connected \MCP~\cite{VandinUR2011}, and the generalized \MCP~\cite{CohenK08}. These variants share the common objective of selecting sets to maximize the total weight of elements in the union set.

 However, it is worth noting that our \MCQC has a unique challenge: finding the datasets that are directly or indirectly connected to the query and maximizing the union of elements contained by the found datasets and query. The existing solutions are designed for different constraints and cannot be used directly to solve our problem. Secondly, among the existing solutions, greedy algorithms are one of the most efficient ways. However, they need to traverse all candidate datasets and perform the exact computation, which is very time-consuming. Therefore, efficient indexes are urgently needed to speed up the search process.

\section{Definitions}
\label{sec:defnitions}


In this section, we introduce the data modelling and formalize problems of overlap and coverage joinable search. Frequently used notations are summarized in Appendix~\ref{appendix:notations}.
%




\subsection{\textbf{Data Modelling}}
\begin{definition1}
\textbf{(Spatial Point)} A spatial point $p=(x, y)$ is a 2-dimension vector, which has a longitude $x$ and a latitude $y$ {(e.g., p=(116.36422$^\circ$, 39.88781$^\circ$)).}
\end{definition1}

\begin{definition1}
\textbf{(Spatial Dataset)} A spatial dataset $D$ contains a set of $2$-dimension points, $i.e., D= \{p_1, p_2, \cdots, p_{|D|}\}$, where $|D|$ is the size of the dataset $D$.
\end{definition1}

\begin{definition1}
\textbf{(Spatial Data Source)} A spatial data source $\mathcal{D}$ contains a set of datasets, i.e., $\mathcal{D} = \{D_1, D_2, \dots, D_n\}$, where $n$ denotes the number of datasets in the data source.
\end{definition1}



As we introduced in Section~\ref{sec:intro}, for ease of quantifying the degree of overlap between the two datasets and the coverage of the datasets, we used a grid partitioning method to represent the spatial dataset, which is widely used in large spatial data processing~\cite{Aji2013,Eldawy2015,Zhang2009,li2024privacy}.  

\begin{definition1}
\label{def:cell}
\textbf{(Grid and Cell)} A 2-dimensional space can be divided into a grid $\mathcal{C}_{\theta}$ consisting of $2^\theta \times 2^\theta$ cells, where $\theta$ called resolution. 
Each cell $c$ denotes a unit space in the grid $\mathcal{C}_{\theta}$, and its coordinates $(X, Y)$ can be converted into a non-negative integer to uniquely indicate $c$ using the z-order curve space-filling method~\cite{Peng2016, Yang2022, SieranojaF18}, denoted as $z(X, Y) = c$. The coordinates are transformed into cell (integer) IDs which are consecutive and form the range $[0, 2^\theta \times 2^\theta-1]$. 
\end{definition1}
Based on the grid partition, each spatial dataset can be represented as a finite set consisting of a set of cell IDs. 

\begin{definition1}
\label{def:spatialSet}
\textbf{(Cell-based Dataset)} 
Given a 2-dimensional space and a grid $\mathcal{C}_{\theta}$, a cell-based dataset $S_{D,\mathcal{C}_{\theta}}$ of the dataset $D$ is a set of cell IDs where 
each cell ID denotes that there exists at least one point $(x, y)\in D$ falling into the cell $c$ of the grid $\mathcal{C}_{\theta}$. The coordinates $(X, Y)$ of cell $c$ are $(\frac{x-x_0}{\nu}, \frac{y-y_0}{\mu})$, where $(x_0,y_0)$ represents the coordinate of the bottom-left point in the whole 2-dimensional space, and $\nu$ and $\mu$ denote the width and height of the cell $c$. Thus, the cell-based dataset is represented as $S_{D,\mathcal{C}_{\theta}} = \{ z(\frac{x-x_0}{\nu}, \frac{y-y_0}{\mu}) | (x,y) \in D \}$. For ease of presentation, we omit the subscript $\mathcal{C}_{\theta}$ of $S$ when the context is clear.
 


\end{definition1}




\begin{example1}
\label{example:raster1}
Figs.~\ref{fig:rasterPreprocessing}(a) and (b) show the process of generating the cell-based dataset by using the grid partition method. 
As shown in Fig.~\ref{fig:rasterPreprocessing}(a), the size of data space containing all datasets is $h\times w$. When performing the grid partitioning, we divide the whole space into a $2^{\theta} \times 2^{\theta}$ grid ($\theta=2$). 
Each cell in the grid can be represented by an integer by interleaving the binary representations of its coordinate values, e.g., for the bottom left cell, its ID is $0$, which is transformed from its coordinates (0, 0). 
Thus, each spatial dataset $D$ can be represented as a finite set $S_D$ consisting of a sequence of cell IDs. 
As shown in Fig.~\ref{fig:rasterPreprocessing}(b), three cell-based datasets of $D_1$, $D_2$ and $D_3$ is $S_{D_1}= \{9,11\}$, $S_{D_2}= \{1, 3\}$ and $S_{D_3}= \{12, 13\}$. The spatial coverage of $S_D$ is the number of IDs in the set.
\end{example1}




\begin{figure}[t]
\setlength{\abovecaptionskip}{-0.2 cm}
\setlength{\belowcaptionskip}{0cm}
\centering
\hspace{-0.4cm}
\includegraphics[width=9.1cm, height=3.2cm]{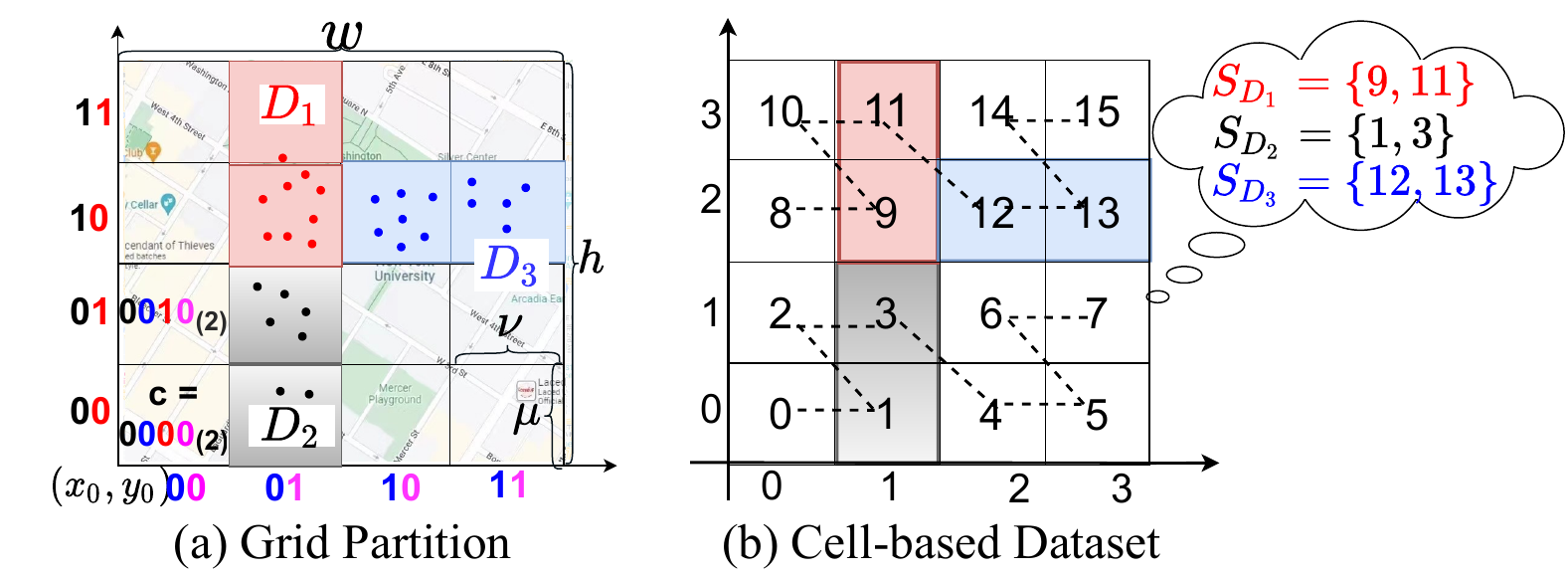}
\caption{Illustration of the grid partition and cell-based dataset.}
\label{fig:rasterPreprocessing}
\end{figure}


Spatial connectivity is a crucial factor in finding relevant datasets, which has also been considered before in trajectory studies~\cite{WangTT2022,wang2021public,Cao2021}, such as the coincidence of the endpoints of trajectories~\cite{WangTT2022}. Similarly, to find a set of datasets that is relevant to the query, we first define a general distance measure to quantify the relationship between two spatial datasets, and then establish criteria defining the connectivity relationships among the datasets.

\begin{definition1}
\label{defi:dist}
\textbf{(Cell-based Dataset Distance)} Given two cell-based datasets $S_D$ and $S_{D'}$, the distance between $S_D$ and $S_{D'}$ is defined as the distance between two nearest cells, i.e.,
\begin{small}
\begin{equation}
    dist(S_D,S_{D'}) = \min_{c_i \in S_{D}, c_j \in S_{D'}}\{ ||c_i, c_j||_2\}.
\end{equation}
\end{small}
where $c_i$ and $c_j$ denote cell IDs, which can be decomposed into coordinates of the cells in the grid, $||\:||_2$ is the Euclidean distance between the coordinates of two cells. 
\end{definition1}

\begin{definition1}
\label{defi:strongConnect}
\textbf{(Directly Connected Relation)} Given a distance threshold $\delta$, we say two cell-based  datasets $S_D$ and $S_{D'}$ are directly connected if $dist(S_D, S_{D'}) \leq \delta$. 
\end{definition1}

\begin{definition1}
\label{defi:weakConnect}
\textbf{(Indirectly Connected Relation)} Two cell-based
datasets $S_D$ and $S_{D'}$ are indirectly connected if exists one or multiple ordered datasets $\{S_{D_1}, \dots, S_{D_i}\} (i \geq 1)$ such that each pair of adjacent datasets (e.g., ($S_{D}, S_{D_1}$), \dots, ($S_{D_i}, S_{D'}$)) are directly connected.
\end{definition1}

\begin{definition1}
\label{defi:connectivity}
\textbf{(Spatial Connectivity)} Given a collection of cell-based datasets $\mathcal{S}_{\mathcal{D}}=\{S_{D_1}, \ldots, S_{D_{|\mathcal{D}|}}\}$, the collection $\mathcal{S}_{\mathcal{D}}$ satisfies the spatial connectivity if and only if any pair of datasets $S_{D_i}$ and $S_{D_j}$ in $\mathcal{S}_{\mathcal{D}} (i \neq j)$ are either directly or indirectly connected relation.
\end{definition1}

\begin{example1}
Following Example~\ref{example:raster1}, we can compute the distance $dist(S_{D_1}, S_{D_2}) = 1$, $dist(S_{D_1}, S_{D_3}) = 1$, and $dist(S_{D_2}, S_{D_3}) = \sqrt{2}$. For a given connectivity threshold $\delta = 1$, the dataset $S_{D_1}$ is directly connected to $S_{D_2}$ and $S_{D_3}$, the dataset $S_{D_2}$ is indirectly connected to $S_{D_3}$. Thus, $S_{D_1}$, $S_{D_2}$ and $S_{D_3}$ satisfy the spatial connectivity.
\end{example1}


\subsection{\textbf{Problem Definitions}}
We define two types of problems to find joinable datasets for the given query. The specific joinable search problems are as follows: 
\begin{definition1}
\textbf{(\underline{O}verlap \underline{J}oinable \underline{S}earch \underline{P}roblem (\textrm{\MIQ}))} Given a cell-based query dataset $S_Q$, a collection of cell-based datasets $\mathcal{S}_\mathcal{D} = \{S_{D_1}, \dots, S_{D_{|\mathcal{D}|}}\}$, and a positive integer $k$, \MIQ aims to find a subset $\mathcal{S}^* \subseteq \mathcal{S}_\mathcal{D}$ such that $|\mathcal{S}^*| \leq k$ and $\forall S_{D_i} \in \mathcal{S}^*$ and $\forall S_{D_j} \in \mathcal{S}_\mathcal{D}\backslash\mathcal{S}^*$, it always satisfies $|S_Q\cap S_{D_i}| \geq |S_Q\cap S_{D_j}|$.
\end{definition1}

\begin{definition1}\label{def:MCQC}
\textbf{(\underline{C}overage \underline{J}oinable \underline{S}earch \underline{P}roblem (\textrm{\MCQC}))}
Given a cell-based query dataset $S_Q$, a collection of cell-based datasets $\mathcal{S}_\mathcal{D} = \{S_{D_1}, \dots, S_{D_{|\mathcal{D}|}}\}$, and a positive integer $k$, \MCQC aims to find a subset $\mathcal{S}^* \subseteq\mathcal{S}_\mathcal{D}$ such that: 
\begin{small}
\begin{align}
      \mathcal{S}^* = & \; \arg\max_{\mathcal{S}^* \subseteq \mathcal{S}_\mathcal{D}} |S_Q \cup (\cup_{S_{D_i} \in \mathcal{S}^*} S_{D_i})|, \\ 
   s.t.  &  \;  \nonumber  |\mathcal{S}^*| \le k, \\
   & \; \nonumber \mathcal{S}^* \cup \{S_Q\} \text{ satisfy spatial connectivity}.  
\end{align}
\end{small}

\end{definition1}

\subsection{\textbf{NP-Hardness of \textrm{\MCQC}}}
\begin{lemma1}
\label{lemma:nphard}
\MCQC is NP-hard.
\end{lemma1}
We prove that \MCQC is NP-hard in Appendix~\ref{appendix:NPhard} by performing a reduction from the NP-hard maximum coverage problem, owing to page limitations.

\section{Overview of Our Search Framework}
\label{sec:framework}

\begin{figure}[t]
\setlength{\abovecaptionskip}{-0.2 cm}
\centering
\includegraphics[width=7.3cm]{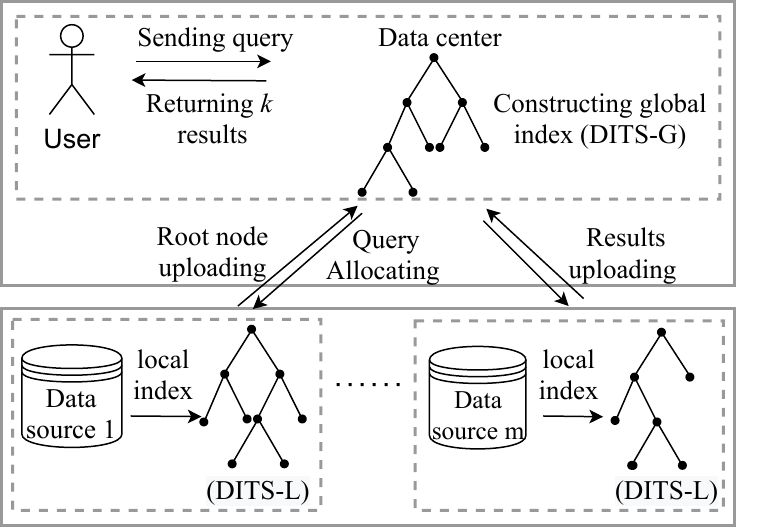}
\caption{Multi-source joinable spatial dataset search framework.}
\label{fig:MSDSFramework}
\end{figure}

In this paper, we design a joinable search framework over multiple spatial data sources, which not only effectively accelerates overlap and coverage joinable search simultaneously, but also provides efficient search over multiple independent data sources. As shown in Fig.~\ref{fig:MSDSFramework}, the search framework contains a data center and multiple independent data sources. 
Each data source organizes its data and constructs the local index. The data center receives the distribution information from data sources and maintains a global index. When the data center receives a query request, it first performs a global search and sends the query to the candidate data sources for local search. Then, each candidate data source sends the results back to the data center after completing the local search, and the data center finally performs the aggregation and returns the $k$ datasets to the user. The following provides a high-level description of index construction and search acceleration.


\myparagraph{Index Construction} 
We propose a \textbf{DI}stributed \textbf{T}ree-based \textbf{S}patial index structure (\unifiedIndex), to support joinable search over multiple data sources, where the global index, denoted as \globalIndex, constructed by the data center is used to find relevant data sources that contain possible results, and the local index, denoted as \localIndex, is used to find results in each data source.
Specifically, instead of indexing each dataset, our global index \globalIndex is created by recording the distribution information of data sources, which is mainly used to filter out the data sources irrelevant to the query. Moreover, we propose a novel local index \localIndex, which effectively combines the features of both the balltree and the inverted index. Compared with the inverted and tree indexes, our \localIndex can prune those irrelevant spatial datasets and quickly retrieve the results of two types of joinable search (see Section~\ref{sec:offline}).



\myparagraph{Accelerating Joinable Spatial Dataset Search}
We propose effective query distribution strategies and search algorithms to support two types of joinable searches. Firstly, to finish the multi-source joinable search, the process involves communication between the data source and the data center. Thus, we design two query distribution strategies to reduce communication costs by reducing the frequency of communication and the number of bytes transferred.
Moreover, for the overlap joinable search, we design an efficient filter-verification algorithm to solve \MIQ, where the filter step combines lower and upper bounds and branch-and-bound strategies to prune plenty of dissimilar pairs and get a set of candidates, and the verification step utilizes effective techniques to verify the candidates.
Furthermore, for the coverage joinable search, we propose a heuristic greedy algorithm to solve \MCQC (see Section~\ref{sec:static}).

\section{Index Construction}
\label{sec:offline}

\begin{figure*}
\setlength{\abovecaptionskip}{-0.2cm}\setlength{\belowcaptionskip}{-0.6 cm}
\centering
\includegraphics[width=17cm, height = 3.1cm]{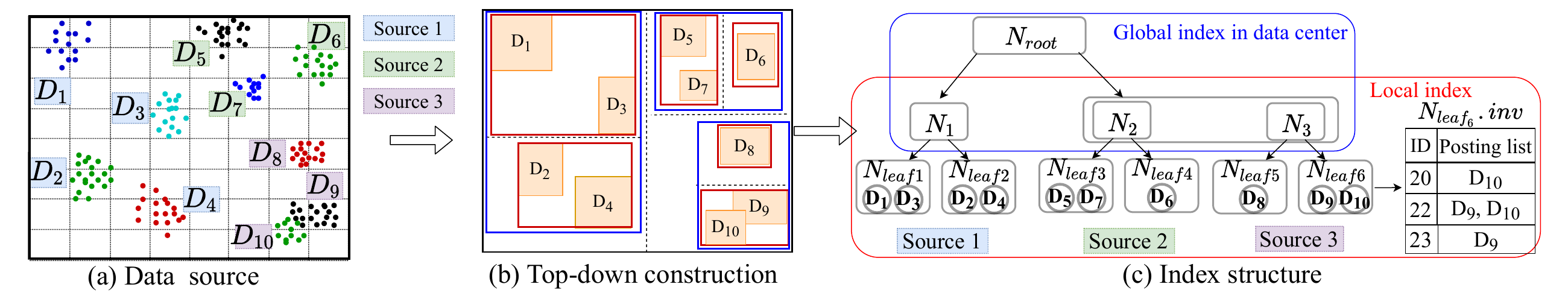}
\caption{An overview of our distributed tree-based spatial index structure.}
\label{fig:twoLevelIndex}
\vspace{-0.7cm}
\end{figure*}


In this section, we focus on constructing an efficient index to support two types of joinable searches over multiple spatial data sources. In our scenario, each data source is independent and autonomous. Thus, we propose a novel distributed tree-based spatial index structure (\unifiedIndex) for spatial joinable search, which is composed of one centralized global index and distributed local indices--one per data source. In what follows, we will introduce the detailed construction process of local and global indexes in Sections~\ref{sec:localTreeIndex} and \ref{sec:globalTreeIndex}, respectively. 


\vspace{-0.2cm}
\subsection{Local Index Construction}
\label{sec:localTreeIndex}
In this subsection, we describe the construction process of the local index. 
For each data source containing a collection of spatial datasets $\mathcal{D} = \{D_1, \dots, D_{n}\}$, we devise a local index to accelerate \MIQ and \MCQC (see Fig.~\ref{fig:twoLevelIndex}). For ease of understanding, we take a single data source as an example to illustrate the local index construction. Before constructing the local index, we first transform each spatial dataset $D$ into a dataset node $N_D$ containing specific information. The formal definition of a dataset node is presented below.

\begin{definition1}
\label{def:datasetNode}
\textbf{(Dataset Node)} A dataset node $N_D= (id, rect, o, r, S_D, pa)$, where $id$ is an dataset identifier, $rect$ is its minimum bounding rectangle (MBR), $o$ is the pivot of the MBR, $r$ is the radius of node, $S_D$ is its cell-based dataset, and $pa$ is the pointer address of the parent node.
\end{definition1}

The $rect$ is a rectangle whose sides are parallel to the $x$ and $y$ axes and minimally enclose all points in the dataset $D$. We compute the pivot $o$ by averaging the coordinates of the bottom left and upper right corners of the $rect$, and choose half of the farthest diagonal distance as the radius $r$. Next, we will give the formal definitions of the internal node and leaf node used in the local index.


\begin{definition1}
\label{def:iternalNode}
\textbf{(Internal Node)} An internal node $N_{in}$ = $(rect, o, r, N_l, N_r, pa)$, where $rect$ is the MBR enclosing all child nodes, $N_l$ and $N_r$ are its left and right child nodes. The meaning of the other symbols in the tuple is the same as in Definition~\ref{def:datasetNode}. The root node $N_{root}$ is an internal node, which has the same structure with $N_{in}$ except without $pa$.
\end{definition1}


\begin{definition1}
\textbf{(Leaf Node)} A leaf node $N_{leaf}$ = $(rect, o, r, ch, inv, f, pa)$, where $ch$ contains its all child dataset nodes, $inv$ is an inverted index of the contained dataset nodes, containing the posting lists that map the cell ID to a list of dataset IDs that contain it, and $f$ is the leaf node's capacity. The other symbols' meaning is the same as in Definition~\ref{def:iternalNode}.
\end{definition1}

\begin{algorithm}[t]
\small
\caption{$\texttt{LocalIndex}(\mathcal{N}_{\mathcal{D}}, pa, f, d)$}
\label{alg:indexConstructure}
\LinesNumbered 
\KwIn{$\mathcal{N}_{\mathcal{D}}$: a list of dataset nodes, 
$pa$: parent node,
  $f$: capacity,
  $d$: dimension}
\KwOut{$N_{root}$: the root node of local index}
   $N_{root} \leftarrow$ generate the root node of $\mathcal{N}_{\mathcal{D}}$ \;\label{ul:genNode2}
   $N_{root}.setParentNode(pa)$\;
   $max\_width \leftarrow -\infty$, $d_{split} \leftarrow 0$; \Comment{split dimension}
   \eIf{$|\mathcal{N}_\mathcal{D}| \leq f$}{ 
    \ForEach{$N_D\in \mathcal{N}_{\mathcal{D}}$}{
       $N_D.setParentNode(N_{root})$\;
       $N_{root}.ch.add(N_D)$\;}
       creating an inverted index for $N_{root}$\;
   }{
   $leftList, rightList\leftarrow \emptyset$;\quad\Comment{represent left and right subsets of $N_{root}$}\label{alg:nodeSpliting}
   \ForEach{$i\leftarrow 1:d$}{\label{line:alg1_traverseD_begin}
   \If{the width of $N_{root}.rect[i] > max\_width$}{
   $max\_width \leftarrow$ the width of $N_{root}.rect[i]$\;
   $d_{split} \leftarrow i$\;\label{line:alg1_traverseD_end}
   }   }
   \ForEach{$N_D\in \mathcal{N}_{\mathcal{D}}$}{ \label{ul:divide1}
   \eIf{$N_D.o[d_{split}] \le N_{root}.o[d_{split}]$}
   {$leftList.add(N_D)$\;}
   {$rightList.add(N_D)$\;}
   }\label{ul:divide2}
   $N_{root}.N_l \leftarrow \texttt{LocalIndex}(leftList, N_{root}, f, d)$\; \label{ul:recall1}
   $N_{root}.N_r \leftarrow \texttt{LocalIndex}(rightList, N_{root}, f, d)$\;\label{ul:recall2}
   
}
\Return $N_{root}$\;
\end{algorithm}


Firstly, we use $n$ to represent the number of datasets in the data source. When constructing the local index, the bottom-up construction needs to repeatedly find the two balls that make the parent node's MBR volume smallest, and insert the parent node into the tree. The total time complexity is $O(n^3)$~\cite{Omohundro1989b}. In contrast, for the top-down construction approach, the algorithm only needs to split the balls into two sets according to the coordinate of the dataset node in the selected dimension, which has a time complexity of $O(n\log{n})$. Thus, we take a top-down approach to construct the local index.

In Algorithm~\ref{alg:indexConstructure}, we show the whole construction process of the local index. Specifically, for a list of dataset nodes $\mathcal{N}_{\mathcal{D}}$ constructed from $\mathcal{D}$, we construct a node $N_{root}$ that contains all dataset nodes as the root node.
Then, we compute the range of coordinate values of the root node in each dimension. We choose the axis with the maximum width as the split dimension and calculate the median of pivot points of all dataset nodes on that dimension, dividing them into left and right subsets based on their position concerning the median (Lines \ref{ul:divide1} to  \ref{ul:divide2}). 
Next, we continue to build the sub-tree for each of the left and right sets of dataset nodes (Lines \ref{ul:recall1} to \ref{ul:recall2}). The construction process continues until the size of the contained dataset nodes does not exceed the leaf node capacity $f$. At this point, we create a leaf node $N_{leaf}$ containing an inverted index.  




\begin{example1}
As shown in Fig.~\ref{fig:twoLevelIndex}(a), there are ten datasets $\{D_1, \dots, D_{10}\}$ in three data sources, and the leaf node capacity $f$ is set to 2.  Then we create the local index for each data source. 
For example, for the 3-th data source, we first transform datasets $\{D_8, D_9, D_{10}\}$ into three dataset nodes in Fig.~\ref{fig:twoLevelIndex}(b). Then, we construct the root node $N_3$ that contains all dataset nodes. We can find the widest dimension and split the dataset nodes into two groups along the $x$-axis. 
Since the number of dataset nodes per group does not exceed $f$, the partitioning process is terminated. The final index structure is shown in the right figure of Fig.~\ref{fig:twoLevelIndex}(c), where the leaf node contains an inverted index that stores a mapping from the cell ID to all dataset IDs.
\end{example1}

\vspace{-0.1cm}
\subsection{Global Index Construction}\label{sec:globalTreeIndex}
When the search scenario includes multiple data sources, one popular solution is to apply a traditional distributed framework with a master-slave structure \cite{Yuan2019, Aji2013, Eldawy2015,ZhangAER2019}. It requires all datasets of slave nodes to be uploaded to the master, then the master node partitions all datasets into equal chunks and stores them separately in the slave node to build indexes of the same setting, such as node capacity, resolution, etc. The partition information held by the master node serves as the global index. However, in our search scenario, each data source is independent. Thus, we propose a global index based on the local indexes, which not only plays the role of an “overview” index but is also suitable for situations where each data source builds its index in different settings.

The global index contains root nodes from local indexes and is used when sending a query to candidate data sources.
Thus, after the local index is built, each source sends its root node $N_{root}$ to the data center. Then, the data center converts the pivot and MBR coordinates of $N_{root}$ into latitude and longitude to resolve the situation where the resolution of the data sources differs. Next, the data center generates the global root node containing the region of all local root nodes and chooses a dimension to split it into two child sets until it does not exceed the leaf node capacity $f$.
We use an example to illustrate the construction process of the global index.

%

\begin{example1}
As shown in Fig.~\ref{fig:twoLevelIndex}(c), each data source sends its root node to the data center. The data center then builds the global index from top to bottom based on the received node information, the construction process is similar to building the local index. The construction process
stops until each leaf node meets the leaf node's capacity, but there is no need to generate the inverted index for leaf nodes. 
\end{example1}

\vspace{-0.1cm}
To deal with the constant update of spatial datasets in each data source, we also design index updating strategies to avoid reconstructing the index. Due to the page limitation, the detailed index update strategies can be found in Appendix~\ref{appendix:indexUpdation}.
Additionally, we theoretically analyze the time and space complexity of constructing \unifiedIndex, which is $O(n\log{n})$ and $O(n)$, respectively. The detailed analysis is provided in Appendix~\ref{appendix:indexComplexity}.



\vspace{-0.1cm}
\section{Accelerating Joinable Search}
\label{sec:static}
In this paper, our goal is to find $k$ datasets from multiple data sources that have the maximum overlap or coverage with the given query dataset.
While the existing distributed queries mainly focus on the range query \cite{Aji2013, Eldawy2015}, k-nearest neighbor (kNN) \cite{Akdogan2010, Eldawy2015}, and approximate nearest neighbor (ANN) \cite{KimP23}, they cannot solve our problem. Thus, we first propose the query distribution strategy to support the efficient communication between the data center and each data source in Section~\ref{sec:queryDistribution}. Secondly, considering the large scale of datasets in each data source and the high computational complexity of the problem, we design efficient dataset search algorithms and acceleration strategies to solve \MIQ and \MCQC in Sections~\ref{sec:MIQsolution} and \ref{sec:MCQCsolution}.


%

\subsection{Query Distribution} 
\label{sec:queryDistribution}
Since the overall search process involves communication between the data source and the data center, as well as local search within the data source, communication cost, and search efficiency are two main bottlenecks in our search framework. Thus, we design two query distribution strategies to solve the above problems. The first strategy is to reduce the frequency of communication. Based on the global index, we can quickly find candidate data sources that may contain query results and transfer the query request only to them. 

Specifically, when a query comes in, we recursively search the global index starting from the root node. For a tree node $N$, we can prune the tree node $N$ directly if it does not overlap with the MBR of the query node ($N.rect \cap N_Q.rect = \emptyset$), or the shortest possible distance between it and the nearest neighbor cell of the query node is greater than the connectivity threshold ($dist(N.o, N_Q.o)-N.r-N_Q.r \geq \delta$); otherwise, we perform a depth-first traversal on it until it is a leaf node. All leaf nodes that intersect with $N_Q.rect$ or possibly connected to $N_Q$ form the candidate data sources. The data center transmits the query only to candidate sources, which reduces the communication cost of query distribution.

The second strategy is to reduce the number of bytes transferred during each communication between the data center and data sources. Because intersections only exist in the area where the MBR intersects, it is unnecessary to transmit the entire $S_Q$ to each candidate source when performing the static search. The data center only transfers the portion of the query that has an MBR intersect to the candidate source to implement the local search to further reduce the communication cost.

\subsection{Accelerating Overlap Joinable Search}
\label{sec:MIQsolution}
The overlap joinable search mainly focuses on finding $k$ datasets with the greatest number of intersections with the query. In this section, we propose an algorithm called \overlapSearch based on the local index to accelerate \MIQ. Since the intersection number of two datasets whose MBR does not intersect is 0, we can directly filter out those tree nodes that do not intersect with the query's MBR. For those nodes intersecting with the query node, we propose lower and upper bounds based on the leaf node's posting list.

\begin{lemma1}\label{lemma:UB}
\textbf{(Upper Bound)} Let $S_Q =\{c_1, c_2, \dots, c_n\}$ represent the cell-based set representation of a query node, the upper bound of intersection between leaf node $N_{leaf}$ and query node $N_Q$ is $\sum_{i=1}^n\varphi(c_i)$.
\begin{small}
$$\varphi(c_i) = \left\{ 
\begin{aligned}
1, & & if\ c_i \in N_{leaf}.inv, \\
0, & &  otherwise
\end{aligned}
\right.
$$
\end{small}
\end{lemma1}
\begin{proof}
As shown in Fig.~\ref{fig:threeBounds}, for a cell ID $c_i$ of $N_Q$, if it is also stored by the key set of the leaf node, which indicates that at least one dataset intersects with $N_Q$.
Therefore, the total number of cells that intersect with the key set of the inverted file is the upper bound between $N_Q$ and $N_{leaf}$.  
\end{proof}

\begin{lemma1}\label{lemma:LB}
\textbf{(Lower Bound)} The lower bound of intersection between $N_{leaf}$ and $N_Q$ is $\sum_{i=1}^n\varphi(c_i)$.
\begin{small}
    $$\varphi(c_i) = \left\{ 
\begin{aligned}
1, & & if\ c_i \in N_{leaf}.inv \& |c_i.pl| = |N_{leaf}.ch|, \\
0, & &  otherwise
\end{aligned}
\right.
$$
\end{small}
\end{lemma1}

\begin{figure}[t]
\setlength{\abovecaptionskip}{-0.2 cm}
\setlength{\belowcaptionskip}{-0.1 cm}
\centering
\hspace{-0.6cm}
\includegraphics[width=8cm]{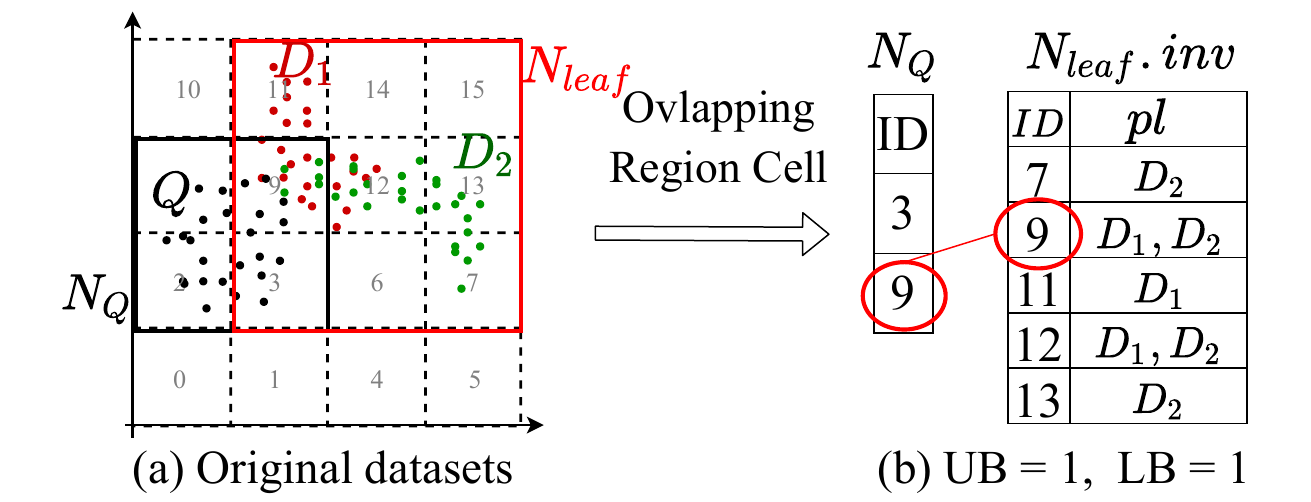}
\caption{Upper and lower bounds of the set intersection.}
\label{fig:threeBounds}
\end{figure}

\begin{proof}
From Fig.~\ref{fig:threeBounds} we can observe that the posting list of $N_{leaf}.inv$ stores the mapping from the cell ID to the spatial dataset. If the size of the $c_i.pl$ in $N_{leaf}.inv$ is equal to the current leaf node's capacity $|N_{leaf}.ch|$, then all dataset nodes contained by the $N_{leaf}$ must contain cell ID $c_i$. Therefore, the lower bound between the leaf node and the query node is the total number of cell IDs whose corresponding list size is equal to the current leaf node's capacity $|N_{leaf}.ch|$. 
\end{proof}

The pseudocode of \overlapSearch can be shown in Algorithm~\ref{alg:MIQSearch}.
Firstly, for a given query node $N_Q$, we implement a recursive search down \localIndex in which we prune away recursive contain internal nodes whose MBR does not overlap the query node (Lines~\ref{alg:NonOverlap1} to \ref{alg:rightNode}). 
Then, for each intersected leaf node $N$, we compute the upper and lower bounds between $N$ and $N_Q$. If the upper bound is less than the current lower bound in the priority queue, all dataset nodes contained by the leaf node can be safely pruned in batch; otherwise, it will become a candidate leaf node (Lines~\ref{alg:computeBound} to \ref{alg:insert}).

After obtaining the first-round candidate nodes, we traverse each candidate leaf node $N_{leaf}$ and compute the set intersection by scanning the posting lists of $N_{leaf}$.
When the size of results queue $\mathcal{R}.size < k$, the candidate dataset is directly inserted into $\mathcal{R}$; otherwise, we compare the set intersection of each candidate dataset with the $k$-th largest value in $\mathcal{R}$ and judge whether to insert it into $\mathcal{R}$. Finally, we can obtain the set of results $\mathcal{R}$ (Lines~\ref{alg:IBS1} to \ref{alg:IBS2}).  Additionally, \overlapSearch incurs $O(n)$ time complexity. The detailed time complexity analysis of the Algorithm~\ref{alg:MIQSearch} can be found in Appendix~\ref{appendix:indexComplexity}.
%

\begin{algorithm}[t]
\setlength{\dbltextfloatsep}{-5cm}
\small
\caption{$\texttt{OverlapSearch}(N_{root},N_Q, k)$}
\label{alg:MIQSearch}
\LinesNumbered 
\KwIn{$N_{root}$: root node of the local index, $N_Q$: query node, $k$: number of results}
\KwOut{$\mathcal{R}$: result queue}
$PQ, \mathcal{R}\leftarrow$ Initialize two priority queues\;
\texttt{BranchAndBound}($N_{root}, N_Q$, $PQ$)\;
\ForEach{$N_{leaf} \in PQ$}{ \label{alg:IBS1}
Compute the exact intersection of all dataset nodes contained by $N_{leaf}$ with $N_Q$ according to $N_{leaf}.inv$\;
\ForEach{$N_D \in N_{leaf}$}{
   \eIf{$\mathcal{R}$.size$\leq k$}{
   $\mathcal{R}$.Insert($N_D$)\;
   }{
   \If{$|N_D \cap N_Q| > \mathcal{R}.peek()$}{
   $\mathcal{R}$.Dequeue()\;
   $\mathcal{R}$.Insert($N_D$)\;
   }
   }
}
}

\Return $\mathcal{R}$\; \label{alg:IBS2}
	\vspace{-0.5em}
	\noindent\rule{8cm}{0.4pt}
  \SetKwFunction{FSum}{BranchAndBound}\label{BranchAndBound}
 \SetKwProg{Fn}{Function}{:}{}
  \Fn{\FSum{$N,N_Q, PQ$}}{
  \KwIn{$N$: tree node, $N_Q$: query node,
  $PQ$: priority queue\;}
  \eIf{$N$ is leaf node}{
  \If{$N.rect \cap N_Q.rect \neq \emptyset$}{
    $N.lb, N.ub \leftarrow$ Compute the lower and upper bounds based on the Lemmas~\ref{lemma:UB} and \ref{lemma:LB}\; 
\label{alg:computeBound}
    \If{$PQ$.isEmpty()}{
    $PQ$.Insert($N$)\;
    }
    \If{$N.ub > PQ.head().lb$}
    {
    \While{$N.lb \geq PQ.head().ub$ and $PQ.size\geq k$}{ \label{alg:dequeue}
    $PQ$.Dequeue()\;
    }
        $PQ$.Insert($N$)\; \label{alg:insert}
    }

  }
  }{
  \If{$N.rect \cap N_Q.rect \neq \emptyset$}
  {\label{alg:NonOverlap1} 
  $\texttt{BranchAndBound} (N.N_l, N_Q, PQ)$ \; \label{alg:leftNode}
    $\texttt{BranchAndBound}(N.N_r, N_Q, PQ)$ \; \label{alg:rightNode}
  }
  }
  } 
\end{algorithm}

\vspace{-0.1cm}
\subsection{Accelerating Coverage Joinable Search}
\label{sec:MCQCsolution}
For the coverage joinable search, we formulate it into an NP-hard \MCQC problem, aiming to find $k$ spatial datasets with the maximum coverage and satisfy the spatial connectivity. It is evident that achieving the optimal solution for \MCQC is computationally prohibitive. In addition, unlike the classic MCP~\cite{Hochbaum1998}, our \MCQC is a query-driven search problem with the spatial connectivity constraint. The basic greedy solution is that we iteratively traverse all datasets in the data source to find datasets directly connected to each dataset in the result set $\mathcal{R}$, and then select the dataset with the maximum marginal gain to add to the result set. 
Thus, the time complexity of each round search is $O(|\mathcal{R}|n)$, where $|\mathcal{R}|$ denotes the number of result datasets and $n$ denotes the number of datasets in the data source.
As the number of result datasets $|\mathcal{R}|$ gradually increases (from $1$ to $k$), the total time complexity is $O(n+2n+\dots+kn)=O(nk(k-1))$. 

To accelerate the search process, we design a greedy algorithm with spatial merge based on the local index, called \coverageSearch. Specifically, we first design efficient upper and lower bounds based on the local index, which greatly accelerates the time of finding the connected datasets. Secondly, we design a merge strategy, which merges the found results into a large dataset for searching, reducing the number of searches and effectively solving \MCQC.  






In solving \MCQC, we define the marginal gain $g(S_D, \mathcal{R})$ as the increased number of cell IDs in the result set $\mathcal{R}$ after adding a new cell-based dataset $S_D$:
\vspace{-0.2cm}
\begin{small}
\begin{equation}
    g(S_D, \mathcal{R}) = |S_D \cup (\cup_{S_{D_i} \in \mathcal{R}} S_{D_i})|- |\cup_{S_{D_i} \in \mathcal{R}} S_{D_i}|.
\end{equation}
\end{small}
The strategy of our proposed greedy algorithm is to iteratively pick the dataset that satisfies the spatial connectivity and has the maximum marginal gain. However, it is expensive to compute the spatial connectivity for each pair of datasets. Thus, we derive lower and upper bounds for distances based on the triangle inequality to accelerate the connectivity verification. 


\begin{figure}[t]
\vspace{-0.5cm}
\setlength{\abovecaptionskip}{-0.2 cm}
\setlength{\belowcaptionskip}{0 cm}
\centering
\includegraphics[width=8.0cm,height=3.6cm]{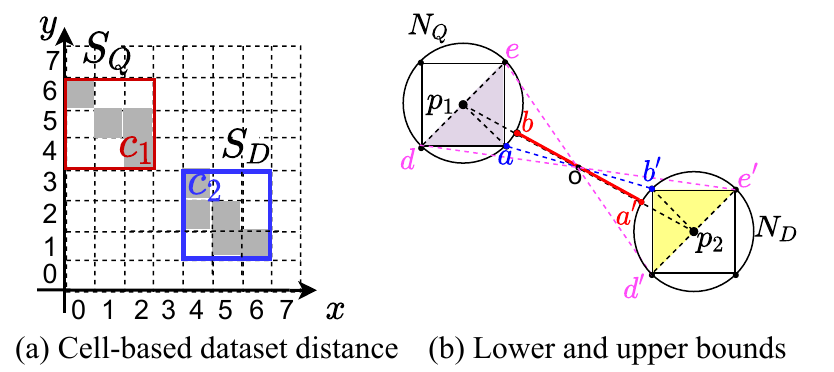}
\caption{Lower and Upper bounds of the cell-based dataset distance.}
\label{fig:bound}
\end{figure}

\begin{lemma1}
\label{lemma:distanceBound} Let $N_Q$ be a node of the query set $S_Q$, $N_D$ be a node of the cell-based dataset $S_D$. The distance $dist(S_Q, S_D)$ between $N_Q$ and $N_D$ can be bounded in the range:
\begin{small}
\begin{equation}
 \begin{aligned}
	 [ \, \max\{ ||N_Q.p, N_D.p||_2 - N_Q.r - N_D.r, 0\},\\ ||N_Q.p, N_D.p||_2 + N_Q.r + N_D.r\, ] \,. \\
	\end{aligned}
	\label{eq:nodeLB_QP}
\end{equation}
\end{small}
\end{lemma1}

\begin{proof}
We first prove the correctness of the lower bound.
As shown in Fig.~\ref{fig:bound}, $N_Q$ and $N_D$ are two nodes transformed by datasets $S_Q$ and $S_D$. Let $a$ and $b$ denote the two closest points in two MBRs, and $a'$ and $b'$ denote two intersect points in the line segment connected by two centers $p_1$ and $p_2$. 
The triangle composed of points $p_1, a$ and $o$ is denoted by $\triangle p_1ao$. According to the triangle inequality, we have:
\begin{small}
\begin{equation}
 \begin{aligned}
||p_1, a||_2 + ||a, o||_2 \geq ||p_1, o||_2.
	\end{aligned}
	\label{eq:nodeUB_QP1}
\end{equation}
\end{small}

Next, based on Inequation~(\ref{eq:nodeUB_QP1}) and $||p_1,a||_2 = ||p_1, b||_2$, we can deduce that $||a, o||_2 \geq ||b, o||_2$. Similarly, in $\triangle p_2b'o$, we can deduce 
$||b', o||_2 \geq ||a', o||_2$. Hence, $||a, b'||_2 \geq ||a', b||_2$.
As the dataset distance $dist(S_Q, S_D)$ is to find the minimum of the distance from an element in $S_D$ to its nearest neighbor element in $S_Q$,  $dist(S_Q, S_D) = ||a, b'||_2 \geq ||b, a'||_2= \max\{ ||N_Q.p, N_D.p||_2 - N_Q.r - N_D.r, 0\}$.

In the following, we prove the correctness of the upper bound. As shown in Fig.~\ref{fig:bound}, each ball is divided into two
half-balls by the diagonal $de$ and $d'e'$, each part should have at least one cell. Thus, the maximum distance between two nodes' nearest neighbor cells is $\max\{||d, e'||_2$, $||d',e||$ \}. Similarly, according to the triangle inequality, we can find that $||d, e'||_2 \leq ||p_1, p_2||_2+||p_1, d||_2+||p_2,e'||$ and $||d', e||_2 \leq ||p_1, p_2||_2+||p_1, e||_2+||p_2,d'||$. Thus, the upper bound of the dataset distance $dist(S_Q, S_D) \leq ||N_Q.p, N_D.p||_2 + N_Q.r + N_D.r.$
\end{proof}
\begin{example1}
As shown in Fig.~\ref{fig:bound}, we can compute the dataset distance between the query dataset $S_Q$ and the cell-based dataset $S_D$ is $dist(S_Q, S_D) = ||q_1, d_1||_2 = \sqrt{5} \approx 2.236$. While based on Lemma~\ref{lemma:distanceBound}, we can compute the lower bound of the dataset distance is $\max\{ 5 - \sqrt{2} -\sqrt{2}, 0\} \approx 2.172 \leq 2.236$, and the upper bound is $ 5 + \sqrt{2} + 
 \sqrt{2} \approx 7.828 \geq 2.236$, which proves the validity of the lower and upper bounds.  
\end{example1}

\begin{algorithm}[tbp]
\small
\caption{$\texttt{CoverageSearch}(N_{root},N_Q, \delta,k)$}
\label{alg:MergeGreedy}
\LinesNumbered 
  \KwIn{$N_{root}$: root node of local index, $N_Q$: query node, $\delta$: connectivity threshold, $k$: number of results}
  \KwOut{$\mathcal{R}$: result set}
  $List \leftarrow \emptyset$, $\mathcal{R} \leftarrow \{N_Q\}$, $N_M \leftarrow N_Q$\;
  
  \While{$\mathcal{R}.size() \leq k$}{
  $\tau \leftarrow -\infty$,
  $N_{best} \leftarrow null$\; \label{alg:MGCG-initial}
       $\texttt{FindConnectSet}(N_{root}, N_M, \delta, List)$\;
    \ForEach{$N_{D} \in List$}{ \label{alg:findBest}
    \If{$|N_{D}.S_D| > \tau$}{ \label{alg:MGCG-upperBound}
    \If{$g(N_D, \mathcal{R}) > \tau$}{ \label{alg:MGCG-verify1}
     $N_{best} \leftarrow N_D$\;
     $\tau \leftarrow g(N_D, \mathcal{R})$\;\label{alg:MGCG-verify2}
        }
        
    }
  }
    $\mathcal{R}.add(N_{best})$\; \label{alg:MGCG-addBest}
    $N_M \leftarrow$ Merge $N_M$ and $N_{best}$\;\label{alg:MGCG-merge}
  }
\Return $\mathcal{R}$\;
	\vspace{-0.5em}
	\noindent\rule{8cm}{0.4pt}
  \SetKwFunction{FSum}{FindConnectSet}\label{FindConnectSet}
 \SetKwProg{Fn}{Function}{:}{}
  \Fn{\FSum{$N, N_Q, \delta, List$}}{
  \KwIn{$N$: tree node, $N_{Q}$: query node, $\delta$: connectivity threshold, List: list of dataset nodes that directly connected to $N_Q$\;}
    $lb \leftarrow \max\{||N.p, N_Q.p||_2 - N.r - N_Q.r, 0\}$\; \label{alg:MGCG-d}
    $ub \leftarrow ||N.p, N_Q.p||_2 + N.r + N_Q.r$\; \label{alg:MGCG-d2}
\eIf{ub$\leq \delta$}{
List.add(all dataset nodes contained by $N$)\;
}{
\If{lb $\leq \delta$}{
\eIf{$N$ is leaf node}{
\ForEach{$N_D \in N.ch$}{
\If{$dist(N, N_D)\leq \delta$}{
List.add($N_D$)\; \label{alg:MGCG-addCandi}
}
}
  
}{
 $\texttt{FindConnectSet}(N.N_l, N_Q, \delta, List)$\;\label{alg:MGCG-traverse1}
  $\texttt{FindConnectSet}(N.N_r, N_Q, \delta, List)$\;\label{alg:MGCG-traverse2}
}
}
}\label{alg:MGCG-findConnected}

  } 
\end{algorithm}

Algorithm~\ref{alg:MergeGreedy} shows the pseudocode of \coverageSearch algorithm based on the lower and upper bounds. 
Specifically, we first initialize a result set $\mathcal{R}$ and put query node $N_Q$ into it. 
Then, we define a merged node $N_M$, which is constructed from the merged set $\cup_{N_i \in \mathcal{R}}$. 
Next, we perform $k$ depth-first search in the local index and find the dataset $S_{D} \in (\mathcal{S}_\mathcal{D} \backslash \mathcal{R})$ that is connected to the $S_i$ and has the maximum marginal gain $g(S_D, \mathcal{R})$ in each iteration (Lines~\ref{alg:MGCG-d} to \ref{alg:MGCG-findConnected}).


For each connected dataset node, we can decide whether to filter based on the maximum marginal gain it can bring.  In each round of searching, we use a symbol $\tau$ to represent the maximum marginal gain found so far. Then, for the dataset node that satisfies $|N_{D}.S_D| > \tau$, we compute the exact number of coverage of each dataset node and judge whether it is the local optimal set (Lines~\ref{alg:findBest} to \ref{alg:MGCG-verify2}). By repeating this step until the whole candidate set is traversed, we can add the set that satisfies the connectivity and has the maximum increment to $\mathcal{R}$. After $k$ iterations, the algorithm terminates.

The \coverageSearch algorithm has a time complexity of $O(kn)$ and provides an approximation ratio of $1-1/e$ under certain assumptions.
Please refer to Appendixes~\ref{appendix:indexComplexity} and ~\ref{appendix:approximation} for details due to the page limitation.

\section{Experiments}
\label{sec:exp}

In Section~\ref{sec:setup}, we first introduce the experimental settings. Then, we evaluate the efficiency of the index construction in Section~\ref{sec:effi-index} and the search performance and communication cost of \MIQ in Section~\ref{sec:effi-MIQ}. Afterwards, we present the search performance and communication cost of \MCQC in Section~\ref{sec:effi-MCQC}.


\begin{table*}[h!]
\setlength{\abovecaptionskip}{0cm}
\setlength{\belowcaptionskip}{0cm}	
	\centering
    \caption{Details of five spatial data sources.}
    \label{tab:dataset}
		\renewcommand\arraystretch{1.2}\addtolength{\tabcolsep}{-1pt}
		\begin{tabular}{ccccc}
			\hline
			\textbf{Data source}    &\textbf{Storage (GB)}  &\textbf{Number of datasets} & \textbf{Number of points} & \textbf{Coordinates range}  \\
			\hline
			\texttt{Baidu-dataset}     &0.18      &6,581    &[3,710,526] &  [($ 87^{\circ}31',19^{\circ}59'$), ($127^{\circ}09',46^{\circ}21'$)]  \\
			 \texttt{BTAA-dataset}    &3.39       &3,204    &[96,788,280]  & [($ -179^{\circ}46',-87^{\circ}70'$), ($179^{\circ}99',71^{\circ}40'$)]  \\ 
			\texttt{NYU-dataset}   &0.89     &1,093   &[15,303,410] & [($-138^{\circ}00',-74^{\circ}01'$), ($56^{\circ}39',83^{\circ}09'$)] \\
			\texttt{Transit-dataset}      &0.21      & 1,967   &[522,461]  &  [($ -77^{\circ}73',36^{\circ}81'$), ($-74^{\circ}53',39^{\circ}78'$)] \\
   			\texttt{UMN-dataset}      &1.34      & 5,453   &[54,417,609]  &  [($ -179^{\circ}14',-14^{\circ}55'$), ($179^{\circ}77',71^{\circ}35'$)] \\
			\hline
		\end{tabular}
		\vspace{-1em}
\end{table*}%

\begin{figure*}[htbp]
\centering
\vspace{-0.2cm}
\setlength{\abovecaptionskip}{-0.6 cm} 
\setlength{\belowcaptionskip}{-0.6 cm} 
\subfigure[]{
\begin{minipage}[t]{0.2\linewidth}
\centering
\includegraphics[width=3.4cm,height=2.8cm]{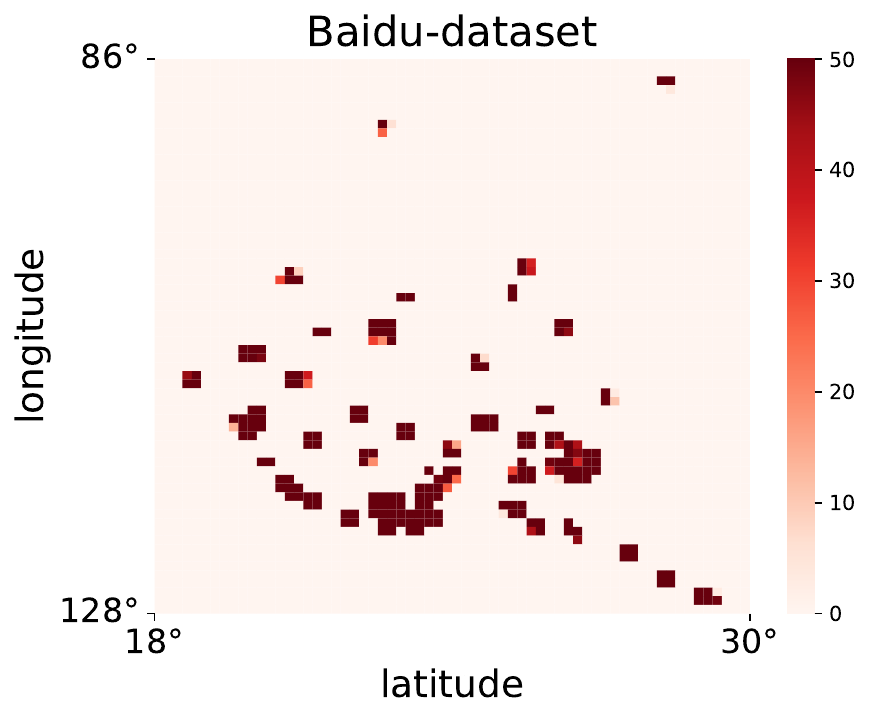}
\end{minipage}%
}%
\subfigure[]{
\begin{minipage}[t]{0.2\linewidth}
\centering
\includegraphics[width=3.4cm,height=2.8cm]{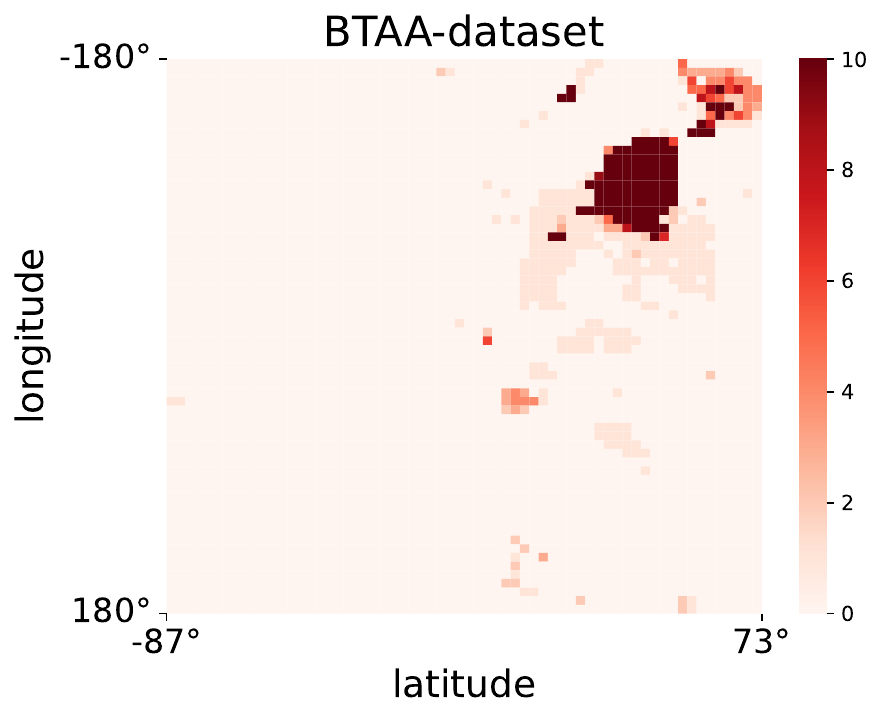}
\end{minipage}%
}%
\subfigure[]{
\begin{minipage}[t]{0.2\linewidth}
\centering
\includegraphics[width=3.4cm,height=2.8cm]{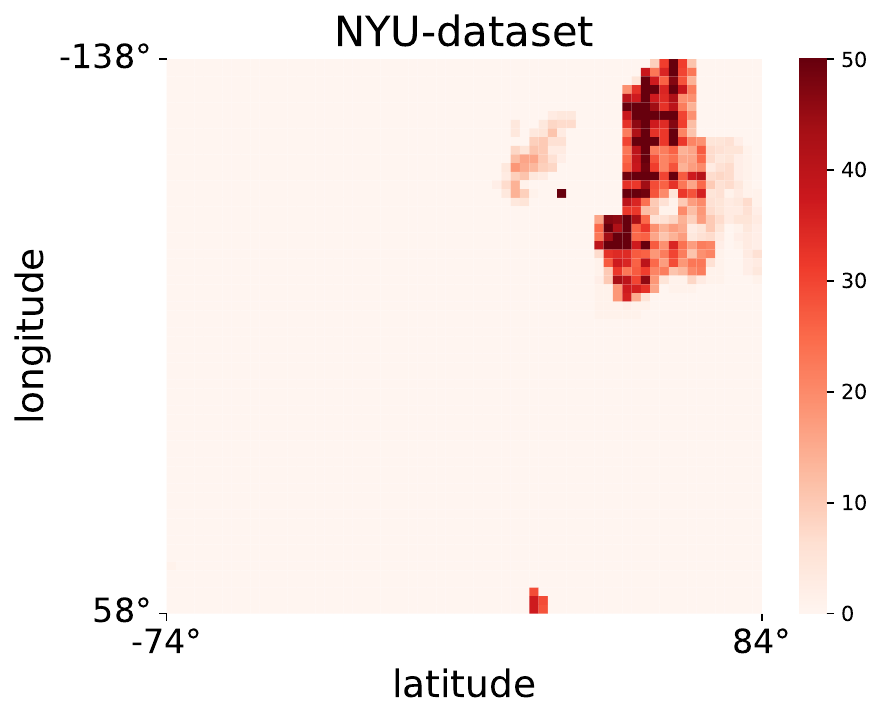}
\end{minipage}%
}%
\subfigure[]{
\begin{minipage}[t]{0.2\linewidth}
\centering
\includegraphics[width=3.4cm,height=2.8cm]{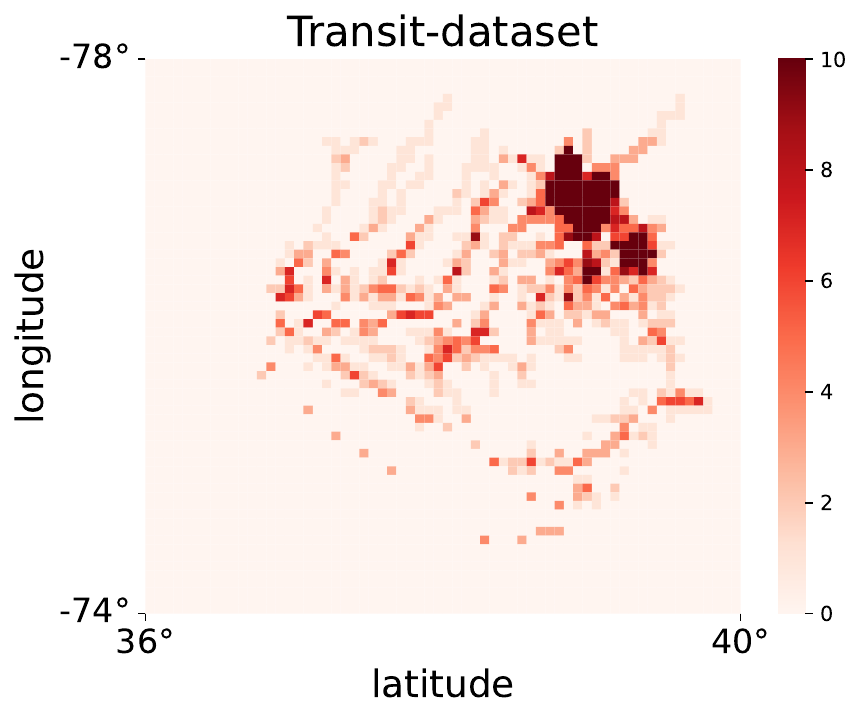}
\end{minipage}%
}%
\subfigure[]{
\begin{minipage}[t]{0.2\linewidth}
\centering
\includegraphics[width=3.4cm,height=2.8cm]{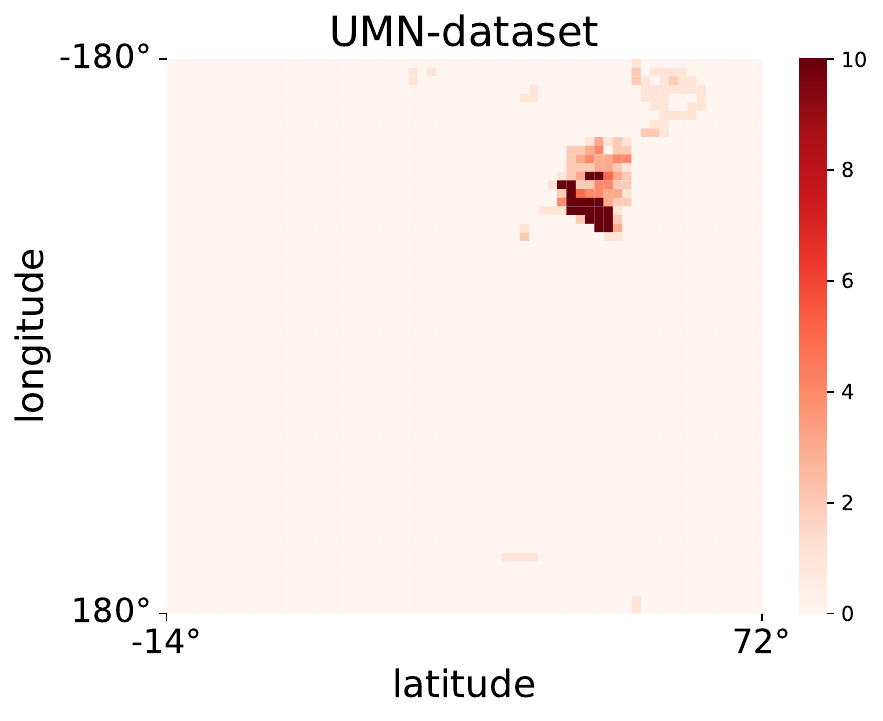}
\end{minipage}%
}%
\centering
\caption{The heatmaps of five data sources.}
\label{fig:heatmap}
\vspace{-0.4cm}
\end{figure*}

\vspace{-0.2cm}
\subsection{Experimental Setups}
\label{sec:setup}
\vspace{-0.2cm}
\myparagraph{Datasets} We conduct experiments on the following five data sources, each of which is downloaded from an open-source spatial data portal. Table~\ref{tab:dataset} shows the statistics of five data sources. Fig.~\ref{fig:heatmap} presents each data source's dataset distribution heatmap, showing the spatial datasets' distribution density in the space. In addition, further details about these five data sources are provided in Appendix~\ref{appendix:dataset}.


\myparagraph{Environment}
To implement the multi-source joinable search on collected five data sources, we conducted all experiments on devices equipped with a 2.90 GHz Intel(R) Core(TM) i5-12400 processor and 16GB of memory. The implementation code can be obtained from the GitHub repository\footnote{\url{https://github.com/whu-totemdb/DITS_code}\label{code}}.

\begin{table}[h]
\setlength{\abovecaptionskip}{0cm}
	\centering
		\caption{Parameter settings.}
  \label{tab:parameter}
  \renewcommand\arraystretch{1.2}\addtolength{\tabcolsep}{-1pt}
		\begin{tabular}{lc}
			\hline
			\textbf{Parameter} & \textbf{Settings}   \\
			\hline
			$k$: number of results  & \{\underline{10}, 20, 30, 40, 50\} \\
			$q$: number of queries &\{ \underline{10}, 20, 30, 40, 50 \} \\
			$\theta$: resolution & \{10, 11, \underline{12}, 13, 14 \}\\ 
            $\delta$: connectivity threshold & \{0, \underline{5}, 10, 15, 20 \}\\
            $f$: leaf node capacity & \{\underline{10}, 20, 30, 40, 50 \}\\
			\hline
		\end{tabular}
\end{table}

\myparagraph{Query and Parameter Generation}
We randomly select 50 datasets from all downloaded datasets as the query datasets. We alter several key parameters including $k$, $q$, $\theta$, $\delta$, $f$ to observe the experimental results. Firstly, the parameter $\theta$ determines the grid's size; thus, we set the resolution according to the distance sampling \cite{Wang2017}. For example, one degree of longitude or latitude is about 111km. If we divide the globe into a $2^{12} \times 2^{12}$ grid, then each cell's area is about $10km \times 5km$. Similarly, the $\delta$ can be set according to the closest distance between the point pairs of two spatial datasets that the user requires. For other parameters, we can also set them according to the requirements of users and data sources. We summarize the parameter settings in Table~\ref{tab:parameter}, and the default value for each parameter is underlined.

\vspace{-0.2cm}
\subsection{Efficiency of Index Constructionn}\label{sec:effi-index}
To accelerate the search performance of \MIQ and \MCQC in each data source, we design the local index \localIndex and show its index construction performance compared with four sota indexes in Fig.~\ref{fig:index1}. Since \localIndex combines the characteristics of the tree index and inverted index, we choose two classical tree indexes and two inverted indexes for comparison.
\noindent
\begin{itemize}
 [leftmargin=*] 
\item \quadtree \cite{Gargantini1982}: \quadtree divides the whole space containing all datasets into four equal-sized quadrants, and then recursively subdivides each quadrant if it contains a sufficient number of data points.
\item  \Rtree \cite{Guttman1984}: \Rtree groups nearby objects and represent them with their minimum bounding rectangle (MBR) in the inner nodes of the tree.
\item \STS \cite{Peng2016}: \STS divides the plane into cells and constructs an inverted index to accelerate the intersection computation.
\item \Josie \cite{zhu2019}: \Josie designs a sorted inverted index where each posting list contains the id, position, and size of datasets.
\end{itemize}

\begin{figure*}[htbp]
\vspace{-0.2cm}
\setlength{\abovecaptionskip}{-0.6 cm} 
\setlength{\belowcaptionskip}{-2 cm} 
\subfigure[]{
\begin{minipage}[t]{0.47\textwidth}
\centering
\includegraphics[width=8.7cm]{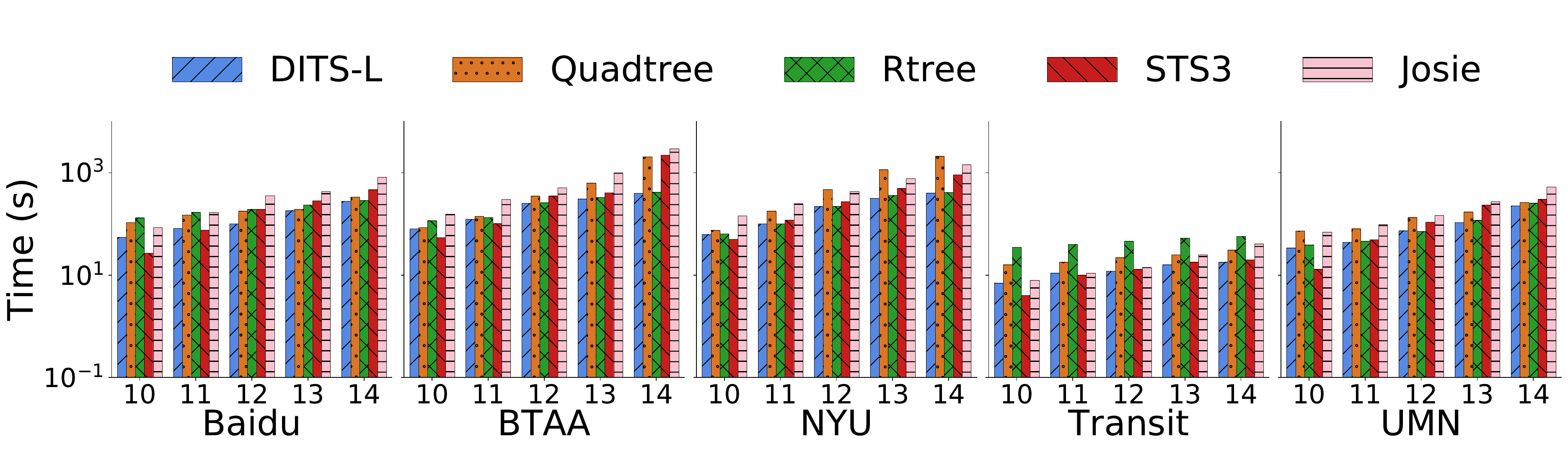}
\end{minipage}%
}%
\subfigure[]{
\begin{minipage}[t]{0.53\textwidth}
\centering
\includegraphics[width=8.7cm]{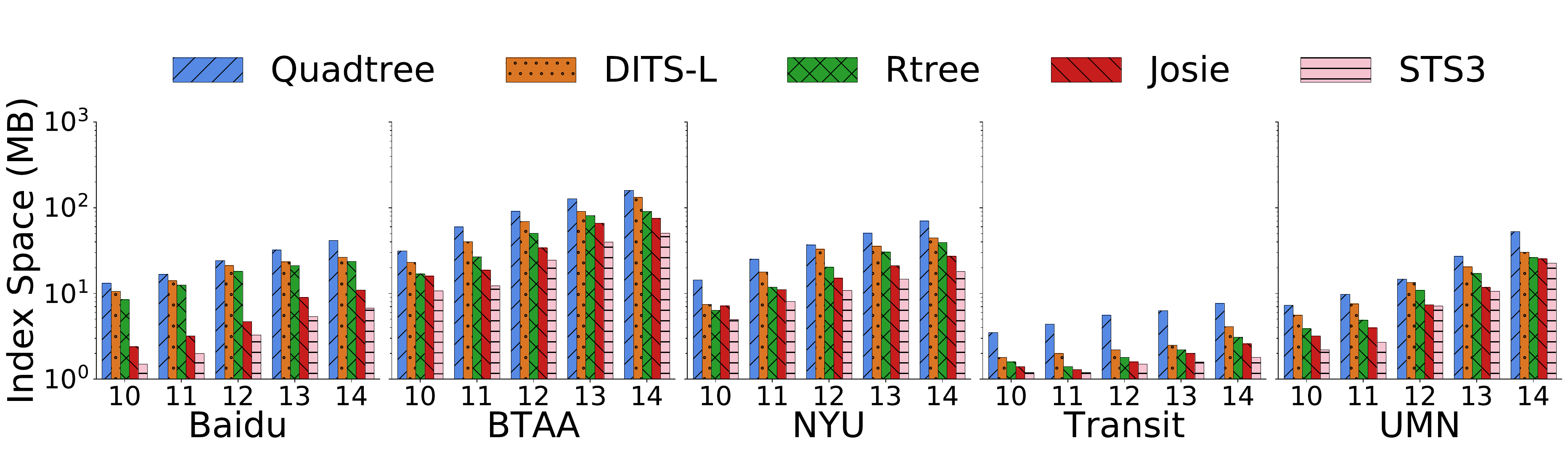}
\end{minipage}%
}%
\centering
\caption{Comparison of five indexes' construction time and memory space with $\theta$ increasing.}\label{fig:index1}
\vspace{-0.6cm}
\end{figure*}

\myparagraph{Time and Space Comparison of Multiple Indexes}
We theoretically compare the index construction time and memory space of \localIndex and four indexes as the $\theta$ increases, as shown in Fig.~\ref{fig:index1}. 
We use $n$ to denote the number of datasets, and $N$ to denote the total number of cell IDs contained by all datasets. The time complexity of \STS~\cite{Peng2016}, \Rtree and \localIndex is $O(n\log{n})$; the time complexity of \Josie\cite{zhu2019} is $O(n^2)$, while \quadtree has a time complexity of $O(N\log{N})$. 
Thus, from the left of Fig.~\ref{fig:index1}, we observe that \Josie is slower than the other indexes in most cases, because it takes more time to sort datasets. Additionally, we can observe that \STS is the fastest at lower $\theta$. This is because when $\theta$ is small, the size of each cell-based dataset is also small. Compared to \STS, \localIndex takes extra time to construct the tree index. Furthermore, we can observe that \localIndex is always slightly faster than \Rtree. This is because \Rtree is a balanced search tree, which needs more time to organize the spatial datasets' MBR.

We also compare the memory space of the five indexes theoretically. Firstly, from the right of Fig.~\ref{fig:index1}, we can observe that the memory space of five indexes gradually increases as $\theta$ increases. This is because a higher resolution results in an increase in the number of cell IDs in each cell-based dataset. In addition, \quadtree occupies the largest memory space, followed by \localIndex, while \STS requires the least memory. This is because \quadtree is built based on the cell IDs of all datasets, and the height of the tree is $O(\log{N})$. In contrast, \localIndex and \Rtree are built based on the dataset, and the height of the tree is $O(\log{n}) (N>>n)$. Thus, \quadtree stores more index nodes compared to \localIndex and \Rtree.

\subsection{Efficiency of Overlap Search and Communication}
\label{sec:effi-MIQ}

\begin{figure*}[htbp]
\setlength{\abovecaptionskip}{-0.3 cm} 
 \begin{minipage}[b]{.48\textwidth}
  \centering
 \includegraphics[width=7.8cm]{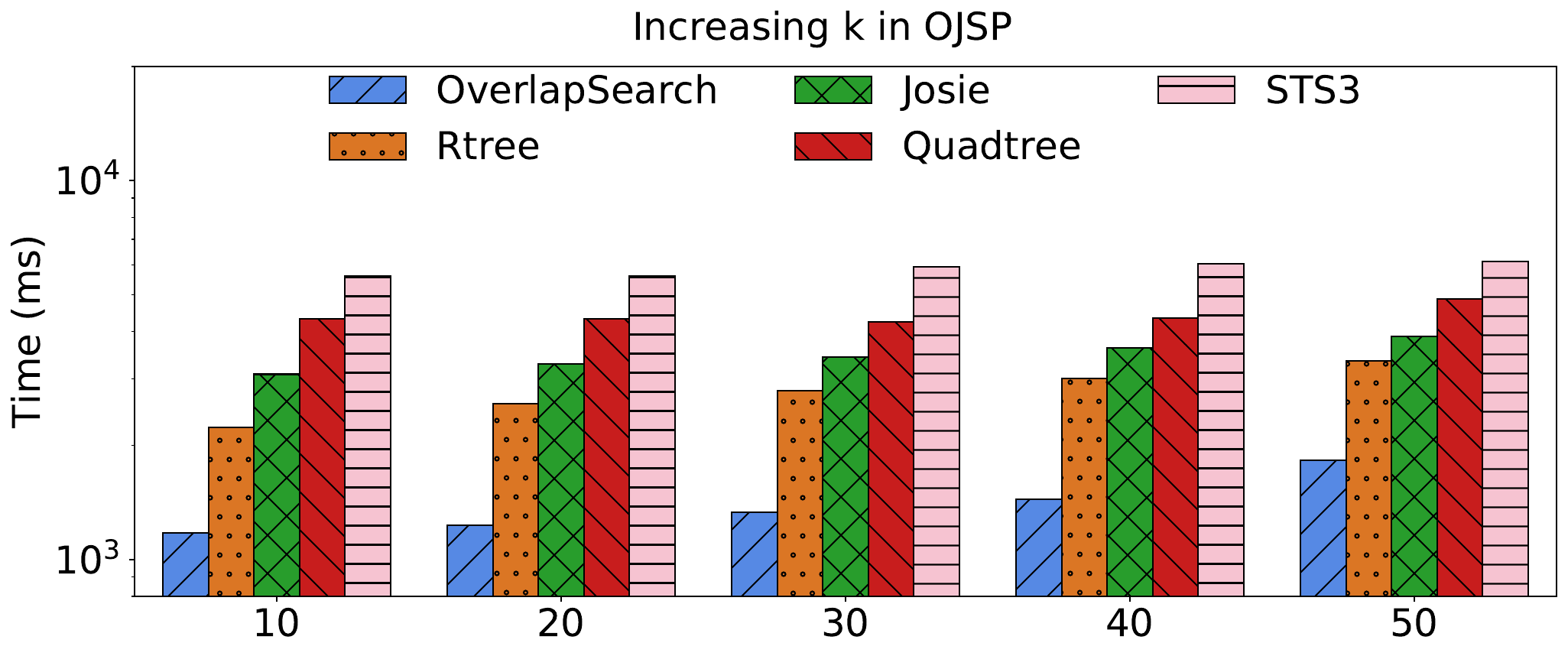}
\caption{Comparison of search time with the increase of $k$.}
\vspace{-1.4em}
\label{fig:topKWI}
 \end{minipage}
 \begin{minipage}[b]{.48\textwidth}
  \centering
  \includegraphics[width=7.8cm]{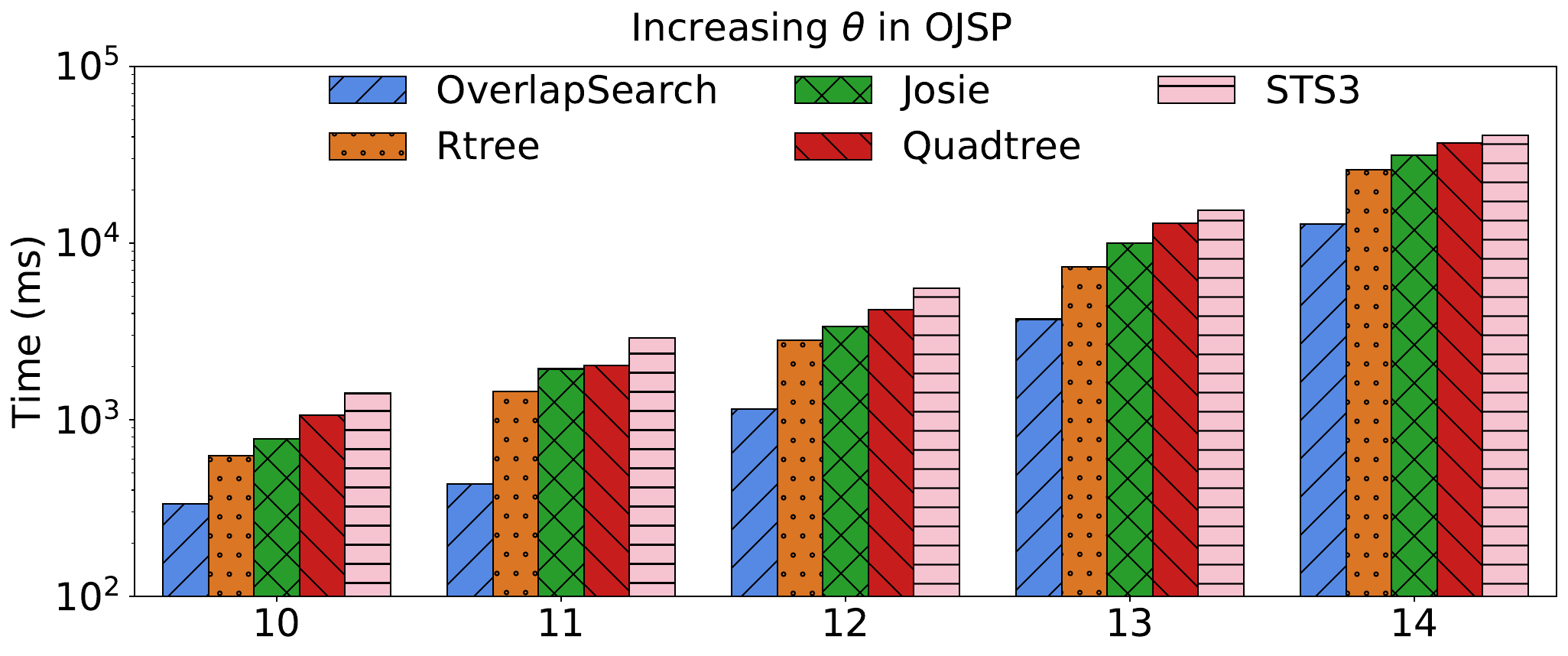}
\caption{Comparison of search time with the increase of $\theta$.}
  \vspace{-1.4em}
\label{fig:topKR}
 \end{minipage}
\end{figure*}

\begin{figure*}[htbp]
\setlength{\abovecaptionskip}{-0.3 cm} 
 \begin{minipage}[b]{.48\textwidth}
  \centering
 \includegraphics[width=7.8cm]{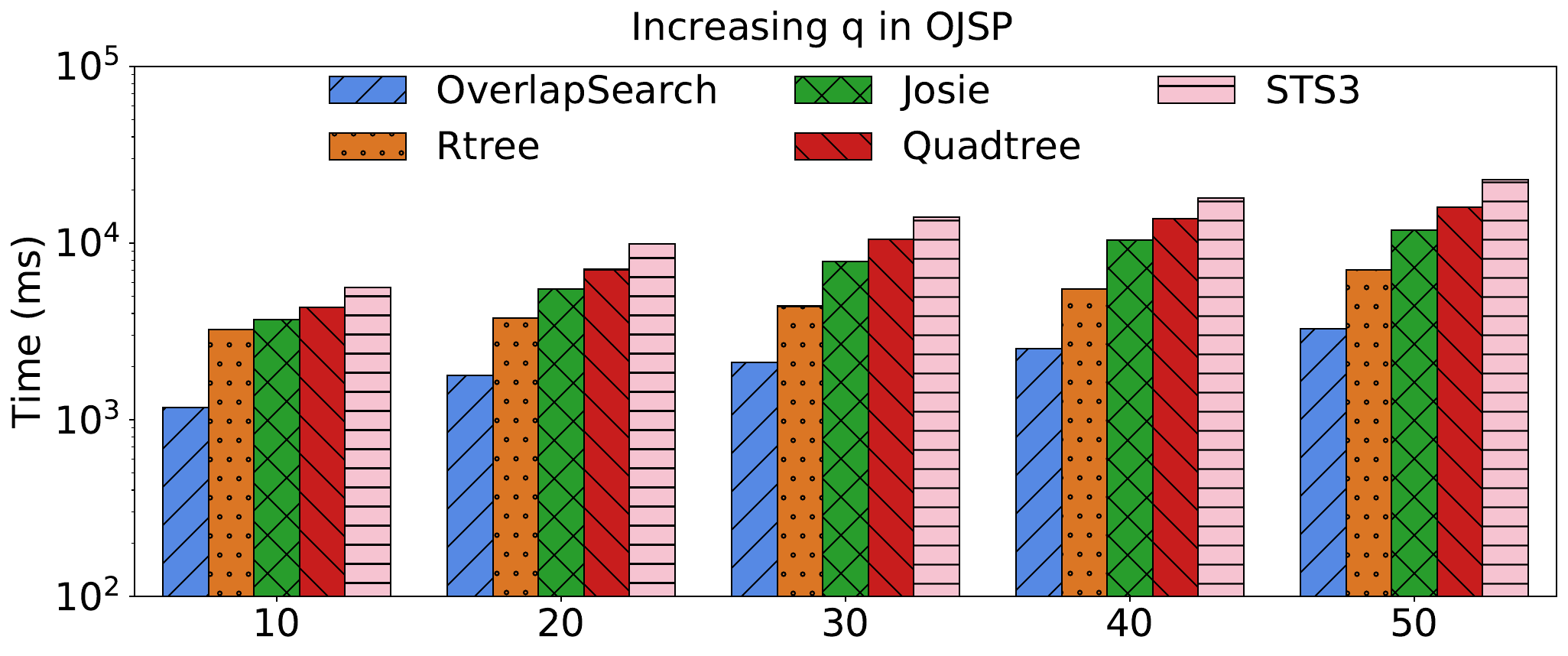}
\caption{Comparison of search time with the increase of $q$.}
  \vspace{-1.4em}
\label{fig:MIPQueryNumber}
 \end{minipage}
  \begin{minipage}[b]{.48\textwidth}
  \centering
 \includegraphics[width=7.8cm]{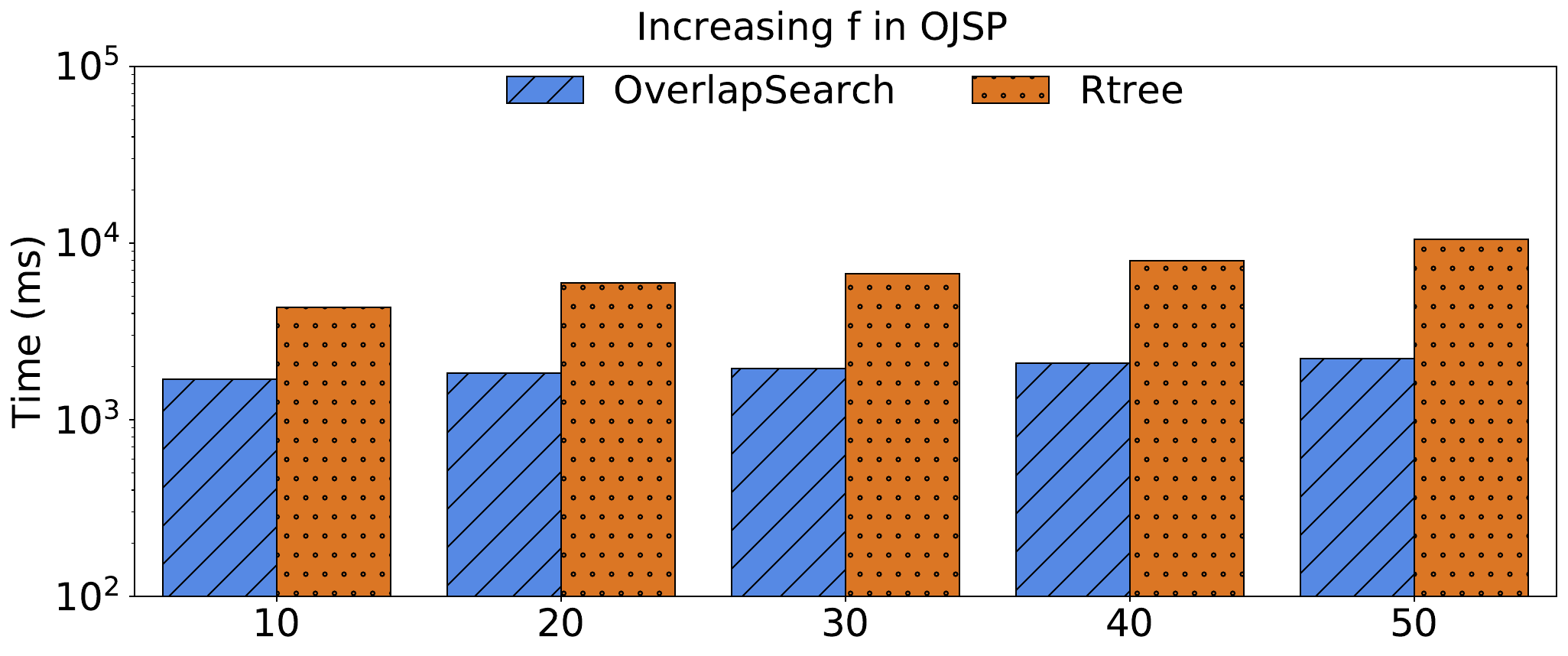}
\caption{Comparison of search time with the increase of $f$.}
  \vspace{-1.4em}
\label{fig:MIPcapacity}
 \end{minipage}
\end{figure*}

\begin{figure*}[htbp]
\setlength{\abovecaptionskip}{-0.3 cm} 
\setlength{\belowcaptionskip}{-0.1 cm} 
 \begin{minipage}[b]{.48\textwidth}
  \centering
  \includegraphics[width=7.8cm]{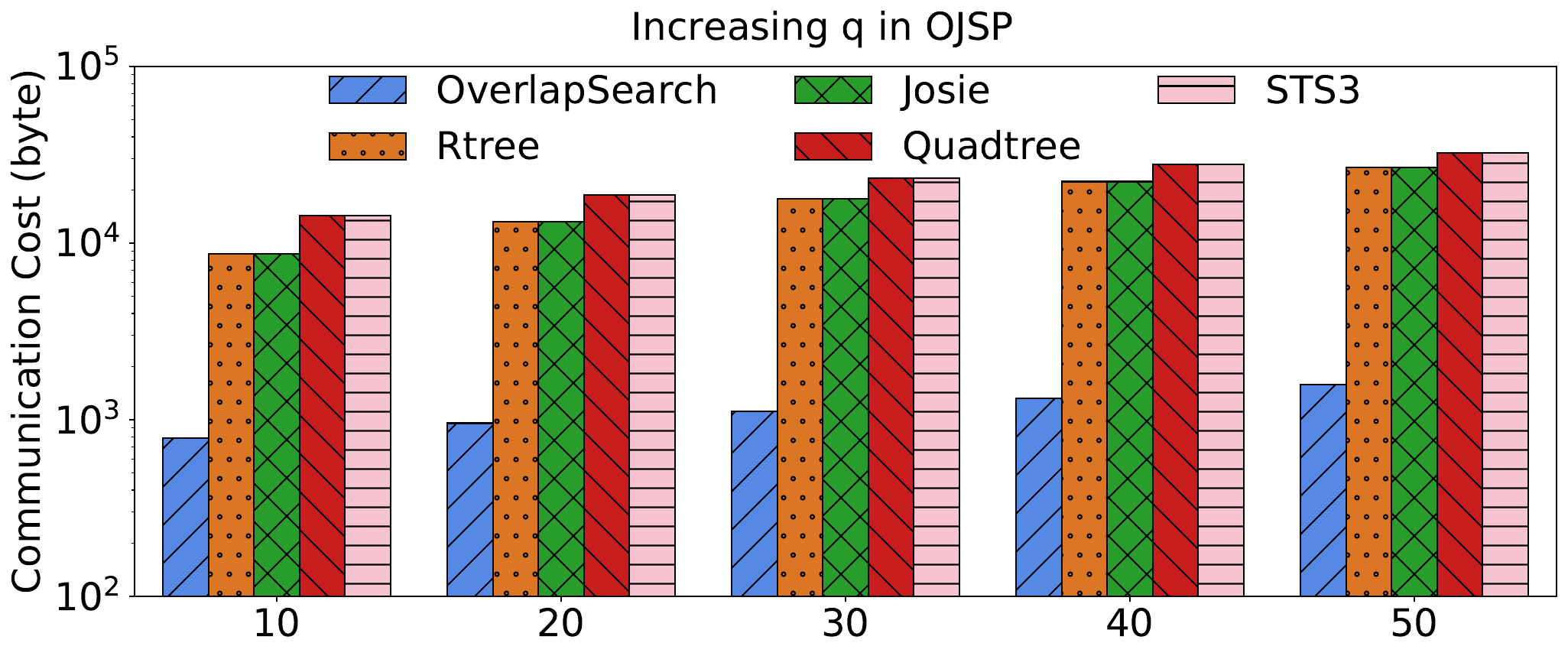}
  
    \caption{Communication cost with the increase of $q$.}
      \vspace{-1.8em}
  \label{fig:Comcost}
 \end{minipage}
 \begin{minipage}[b]{.48\textwidth}
  \centering
  \includegraphics[width=7.8cm]{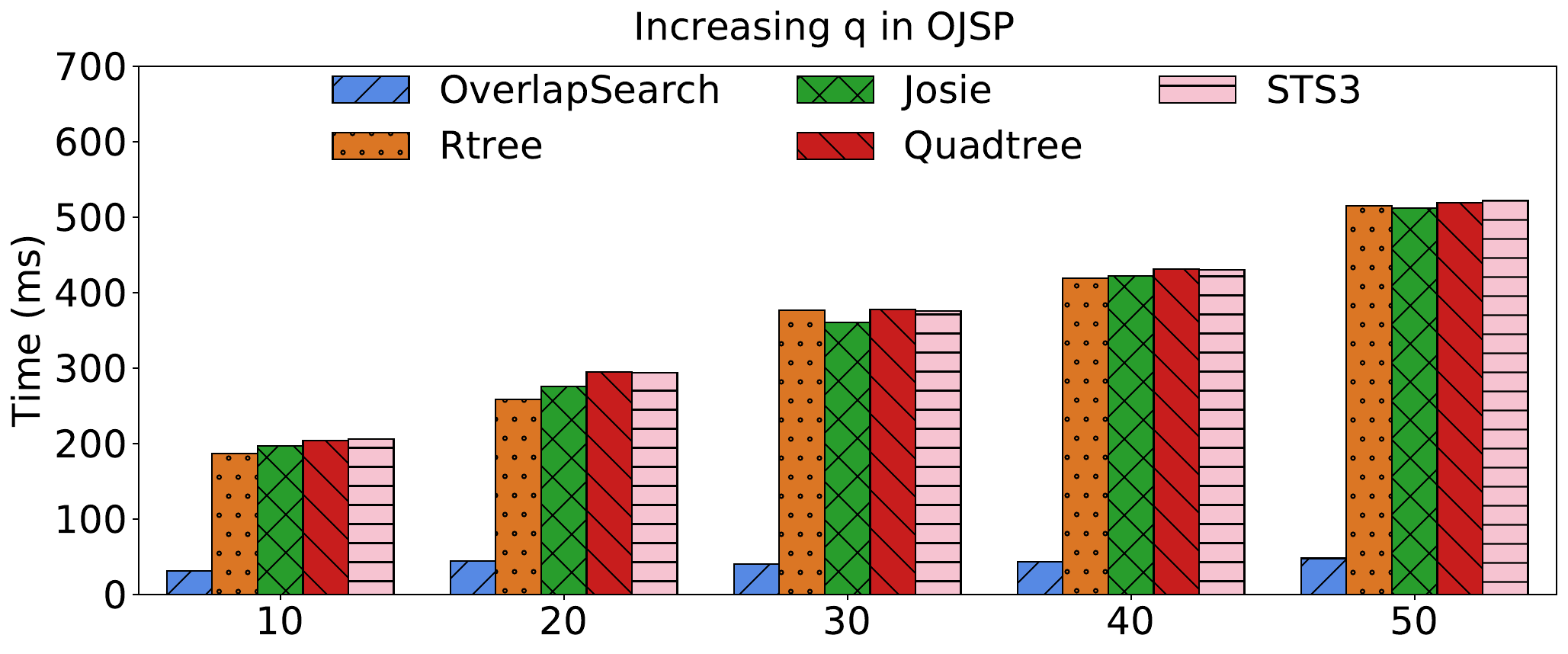}
  \caption{Transmission time with the increase of $q$.}
    \vspace{-1.8em}
 \label{fig:TransmissionTime}
 \end{minipage}
\end{figure*}

\myparagraph{Methods for Comparisons} 
We compare our proposed exact \overlapSearch with four exact algorithms based on the four indexes introduced in Section~\ref{sec:effi-index}. When performing \MIQ based on \quadtree~\cite{Gargantini1982}, we find all leaf nodes that intersect the query and count the dataset IDs of the points that overlap with the query dataset. Moreover, when performing \MIQ based on \Rtree~\cite{Guttman1984}, we find all intersecting datasets based on their MBR, and compute the size of the set intersection. For \STS~\cite{Peng2016} and \Josie~\cite{zhu2019}, we directly apply them to cell-based datasets to compute set intersection.

\subsubsection{Efficiency of Overlap Joinable Search} 
In this section, we present the experimental results of \overlapSearch algorithm and the four comparison algorithms that vary with several key parameters on five data sources.

\myparagraph{Efficiency Comparison with Different $k$}
As shown in Fig.~\ref{fig:topKWI}, we can see that \overlapSearch achieves 1.7 to 4.8 times speedup compared with the other four algorithms, indicating that our distributed index and pruning strategies are effective. 
Although \Rtree and \quadtree also apply the tree index, \Rtree shows the second-best performance compared with \quadtree, as \quadtree is constructed based on the cell IDs of all spatial datasets, we need to find all set elements that intersect with the query and record the sorted results, which is similar to the inverted index. Moreover, we can see that \Josie runs faster than \STS because \Josie can terminate the search early by the prefix filter. We can see that for the search algorithms based on \overlapSearch, \Rtree and \Josie, the runtime shows a slight increase. In contrast, the running time of \quadtree and \STS remains relatively unchanged with increasing $k$. The reason is that they need to sort all datasets that intersect with the query according to the number of overlapping cell IDs. Thus, the change in $k$ has minimal impact on runtime.


\vspace{-0.1cm}
\myparagraph{Efficiency Comparison with Different $\theta$}
We increase $\theta$ to investigate the search performance of \overlapSearch compared with the other four algorithms. We can observe that the search time of the five algorithms increases gradually as the $\theta$ increases in Fig.~\ref{fig:topKR}. The reason is that each spatial dataset is divided into finer granularity as the $\theta$ increases. However, \overlapSearch is always 1.8 to 6.7 times faster than the other four algorithms under different $\theta$. The reason is that our algorithm can quickly filter out unpromising leaf nodes using \localIndex and bounds. In addition, we can also see that the two algorithms based on the tree index maintain a better search performance than the other two based on the inverted index, indicating that filtering based on the geographic characteristics of the spatial dataset can provide better search performance.

\myparagraph{Efficiency Comparison with Different $q$} 
Fig.~\ref{fig:MIPQueryNumber} shows the search performance as the $q$ increases. We can see that the search time of \overlapSearch and \Rtree increases slightly. The running time at $q=50$ increases only by about $3\times$ compared to the running time at $q=10$. In contrast, the other four algorithms have an increase of $4\times$ or more. The reason is that the search based on the tree index has a faster filtering power than the other three algorithms, so it can show more stable search performance as the number of queries increases. In addition, the search based on \quadtree is to find the intersection point, which is similar to the inverted index, thus it has a slower time than the other two tree algorithms.

\vspace{-0.1cm}
\myparagraph{Efficiency Comparison with Different $f$} 
We also investigate the search performance as $f$ increases, and the results are shown in Fig.~\ref{fig:MIPcapacity}. Here, we only show the experiment results of \overlapSearch and \Rtree, as the leaf node capacity in \quadtree is 4, \STS and \Josie are based on the inverted index. We can observe that the running time of \overlapSearch increases slightly as the $f$ increases. This is because bigger $f$ cannot be pruned easily, which degrades the performance. However, our algorithm still has a better performance than \Rtree, as our algorithm has a stronger filtering power and faster computation time based on the lower and upper bounds.

\begin{figure*}[htbp]
\setlength{\abovecaptionskip}{-0.3 cm} 
\setlength{\belowcaptionskip}{-0.3 cm} %
 \begin{minipage}[b]{.48\textwidth}
  \centering
  \includegraphics[width=7.8cm]{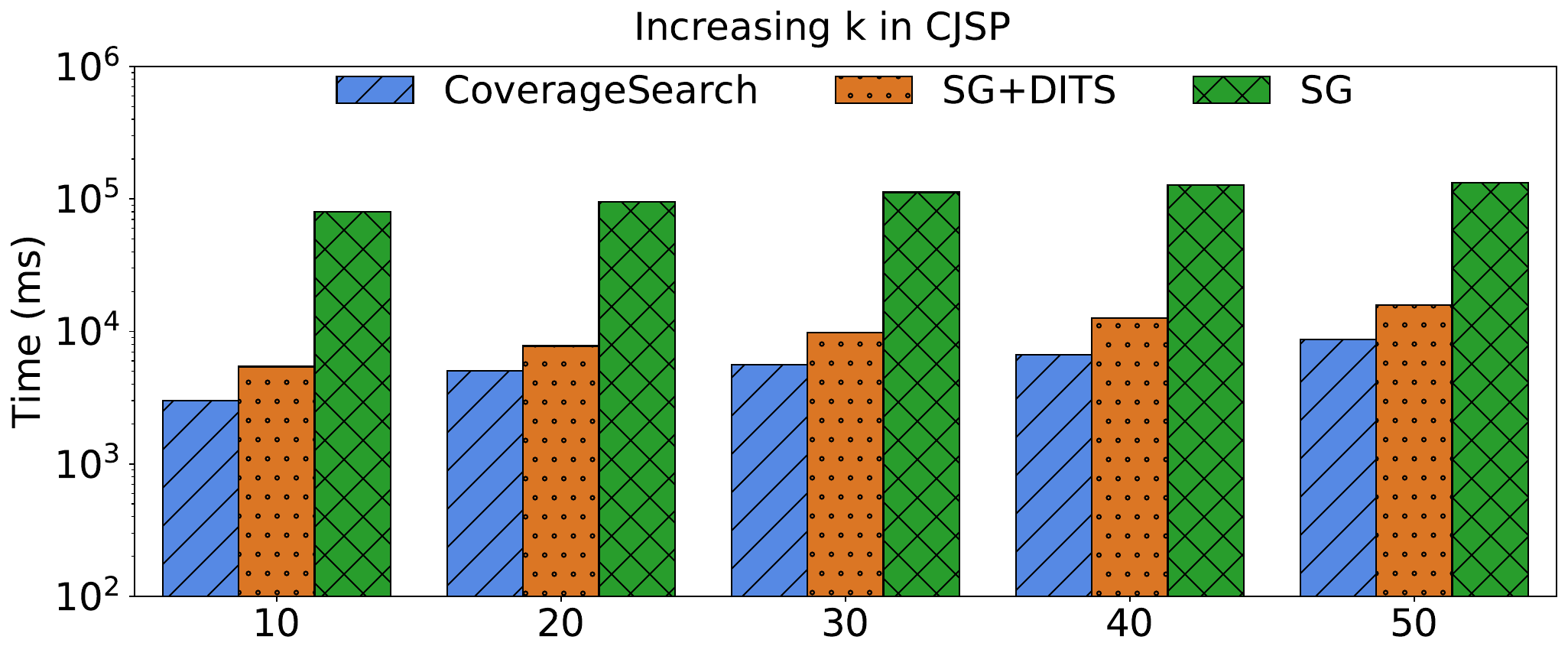}
 \caption{Comparison of search time with the increase of $k$.}
     \vspace{0.3em}
 \label{fig:MCQCtopKWI}
 \end{minipage}
 \begin{minipage}[b]{.48\textwidth}
  \centering
  \includegraphics[width=7.8cm]{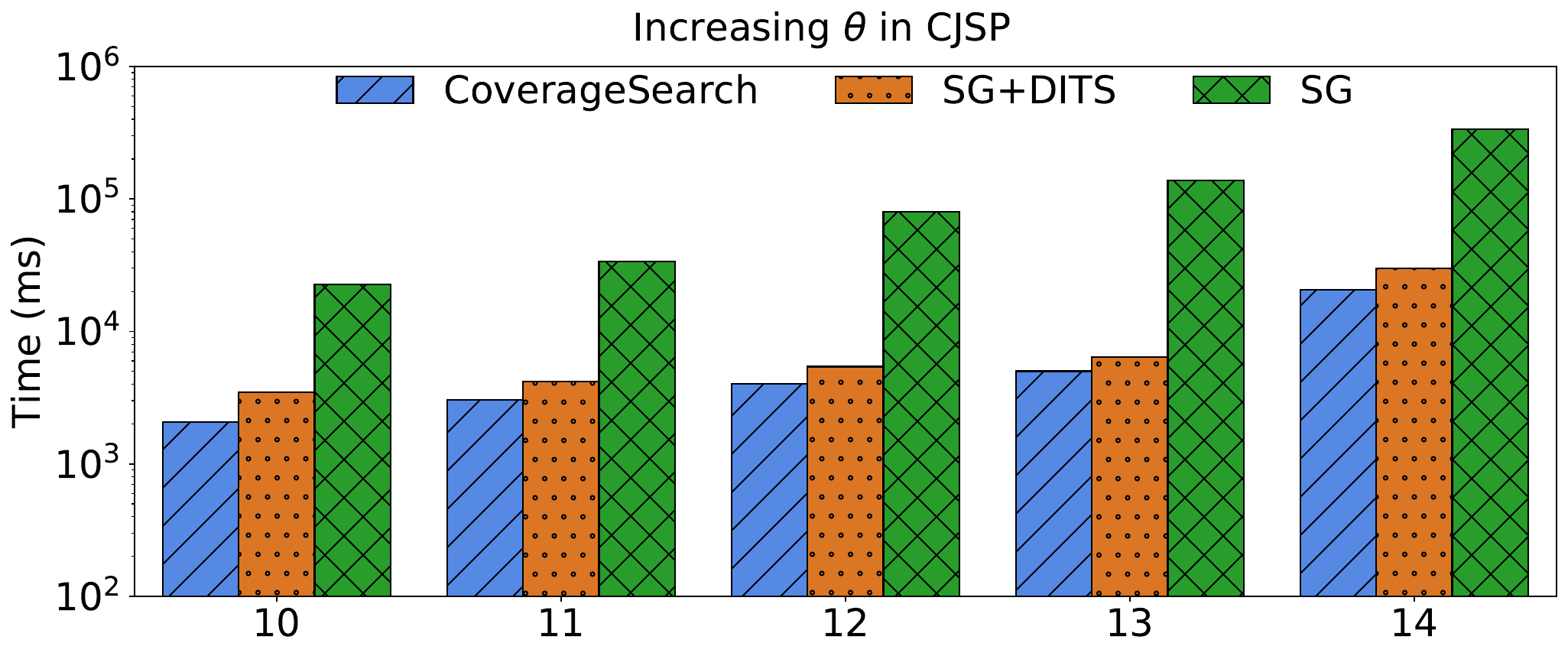}
\caption{Comparison of search time with the increase of $\theta$.}
  \vspace{0.3em}
\label{fig:MCQCtopKR}
 \end{minipage}
 \vspace{-1.5em}
\end{figure*}

\begin{figure*}[htbp]
\setlength{\abovecaptionskip}{-0.3 cm} 
\setlength{\belowcaptionskip}{-0.1 cm} 
 \begin{minipage}[b]{.48\textwidth}
  \centering
  \includegraphics[width=7.8cm]{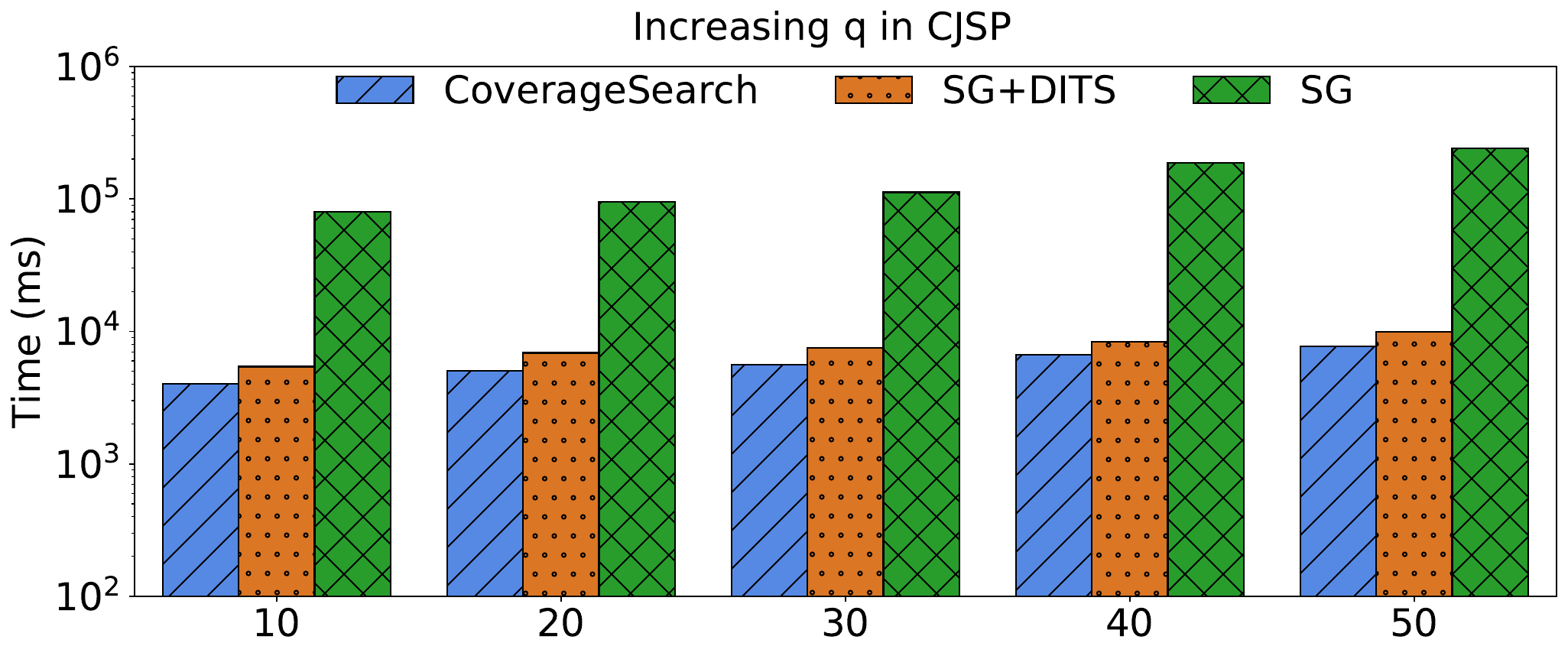}
  \caption{Comparison of search time with the increase of $q$.}
    \vspace{-1.3em}
  \label{fig:MCQCindexConstructed}
 \end{minipage}
 \begin{minipage}[b]{.48\textwidth}
  \centering
  \includegraphics[width=7.8cm]{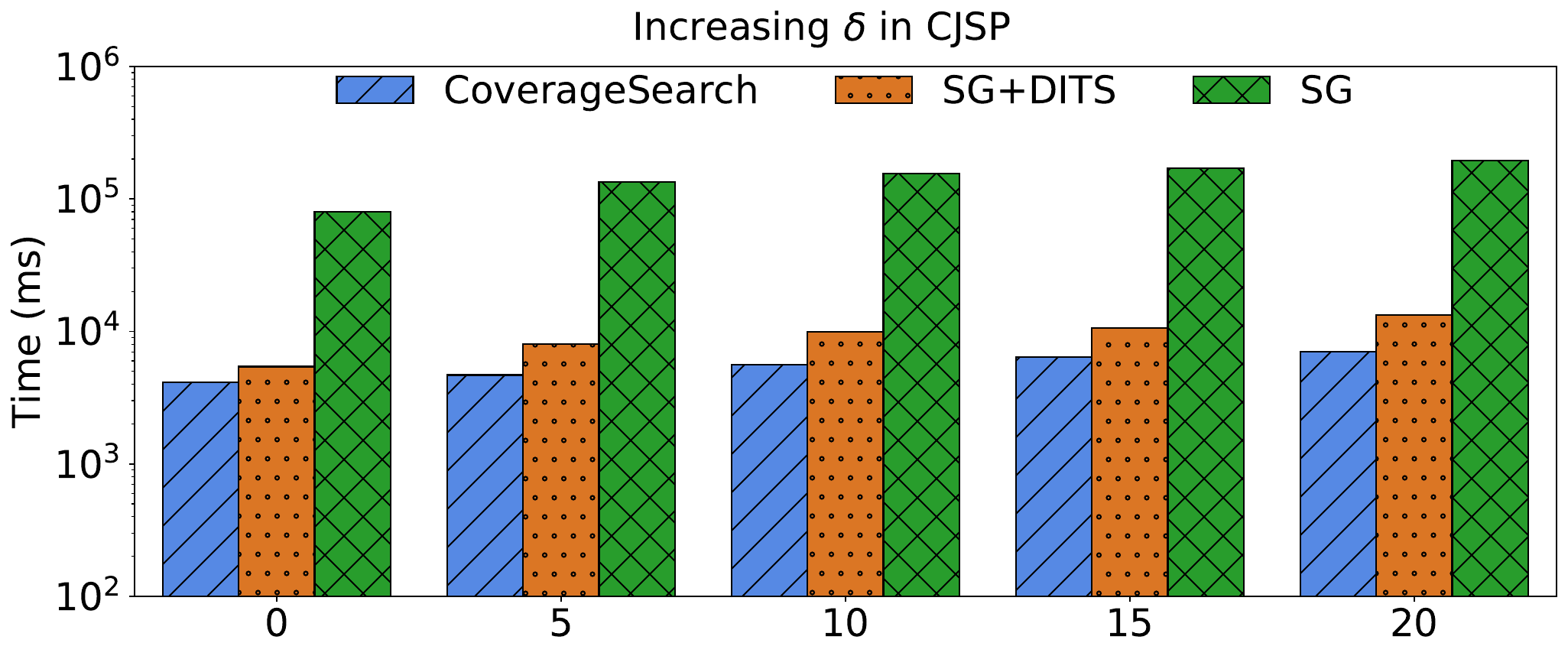}
\caption{Comparison of search time with the increase of $\delta$.}
\vspace{-1.3em}
\label{fig:MCQCdelta}
 \end{minipage}
\end{figure*}

\begin{figure*}[htbp]
\setlength{\abovecaptionskip}{-0.3 cm} 
\setlength{\belowcaptionskip}{-0.2 cm} 
 \begin{minipage}[b]{.48\textwidth}
  \centering
\includegraphics[width=8cm]{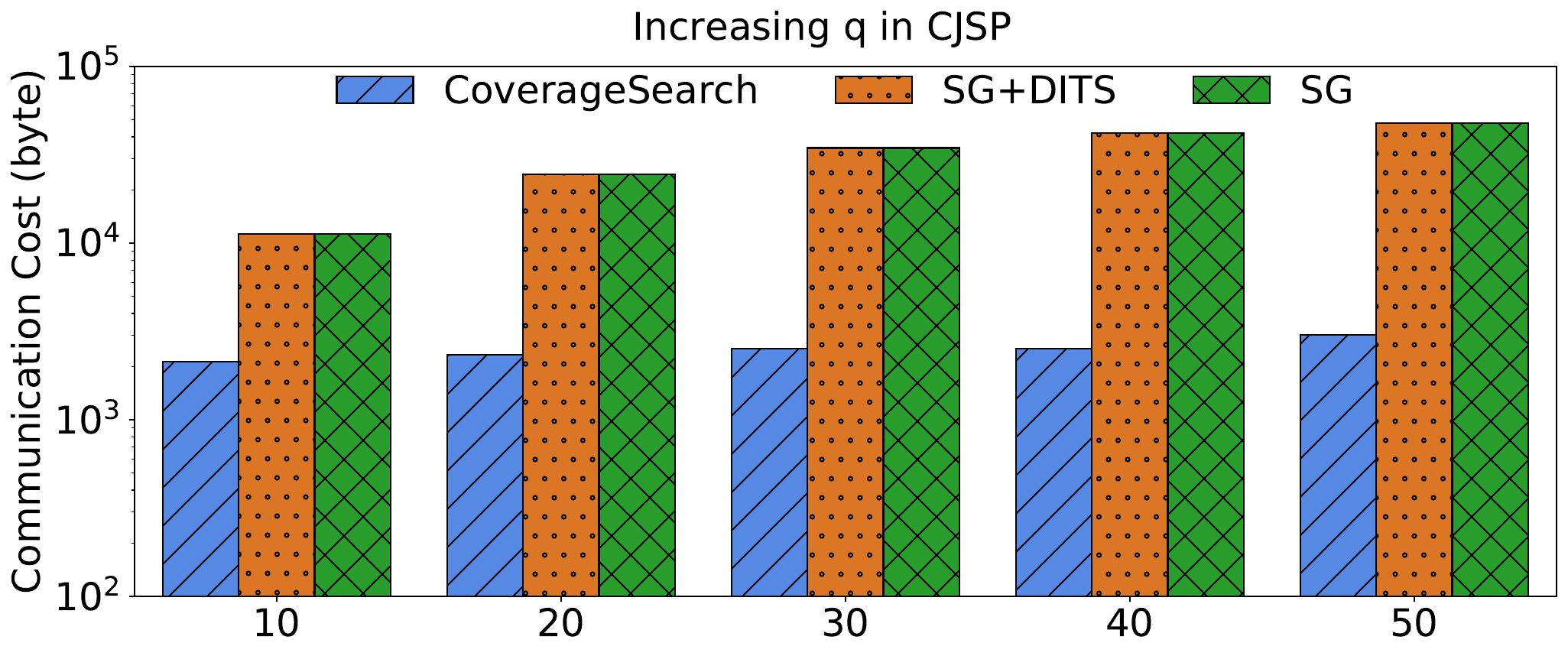}
    \caption{Communication cost with the increase of $q$.}
     \vspace{-1.5em}
  \label{fig:MCQComcost}
 \end{minipage}
 \begin{minipage}[b]{.48\textwidth}
  \centering
  \includegraphics[width=7.8cm]{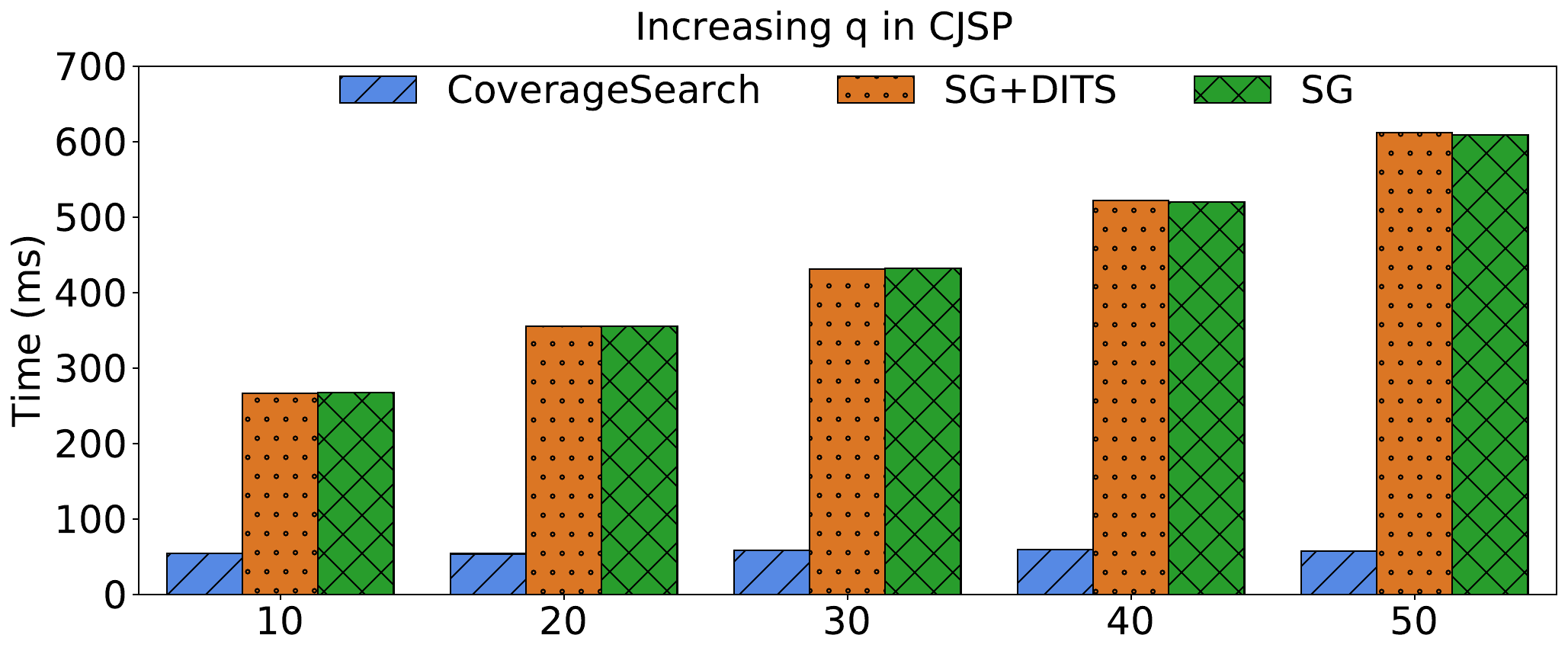}
  \caption{Transmission time with the increase of $q$.}
    \vspace{-1.5em}
 \label{fig:MCPTransmissionTime}
 \end{minipage}

\end{figure*}

\subsubsection{Efficiency of Communication Cost}
We compare the communication cost of \overlapSearch with query distribution strategy compared to other baselines as the $q$ increases.
Fig.~\ref{fig:Comcost} shows the number of bytes transferred as the $q$ increases. We can see that the number of bytes transferred by five algorithms gradually increases as the $q$ increases, wherein \overlapSearch transmits the fewest bytes. This is because \overlapSearch only transmits query region that overlaps with the candidate dataset source. Fig.~\ref{fig:TransmissionTime} shows the transmission time. Since the transmission time is inversely proportional to the network bandwidth when the network bandwidth is constant, the more bytes that are transmitted, the longer the transmission time will be.

\vspace{-0.1cm}
\subsection{Efficiency of Coverage Search and Communication}
\label{sec:effi-MCQC}
\vspace{-0.1cm}
\myparagraph{Methods for Comparisons}
We compare our proposed \coverageSearch algorithm in Section~\ref{sec:MCQCsolution} with two baselines.
\begin{itemize}
 [leftmargin=*] 
    \item \SG~\cite{Hochbaum1998}: 
    It is a standard greedy algorithm for solving \MCP. Here, we extend it to solve our \MCQC, that is, we traverse all datasets in the data source and find the dataset that is directly or indirectly connected with the query and has the maximum marginal gain at each iteration. 
    \item \textsf{SG}+\unifiedIndex: It adopts \localIndex to speed up the search process of finding connected candidate datasets each round, and \globalIndex to find the candidate data sources.
    
\end{itemize}

\subsubsection{Efficiency of Coverage Joinable Search}
We vary key parameters $k$, $\theta$, $q$, and $\delta$ to compare the efficiency of \coverageSearch with two comparison algorithms.

\myparagraph{Efficiency Comparison with Different $k$}
Fig.~\ref{fig:MCQCtopKWI} shows the search time of three algorithms as $k$ increases. We observe that \coverageSearch achieves at most 1.7 and 26.5 times speed up against \textsf{SG}+\unifiedIndex and \textsf{SG}, respectively.  
Because \coverageSearch only needs to search in the index tree once to find the node connected with the query in each iteration, significantly reducing the search time. In addition, we can see that \textsf{SG}+\unifiedIndex is significantly faster than \textsf{SG}, which indicates our local index is very efficient in solving \MCQC. 


\myparagraph{Efficiency Comparison with Different $\theta$}
Fig.~\ref{fig:MCQCtopKR} presents the search time as the $\theta$ increases. We can see that the search time of the three algorithms gradually increases, but \coverageSearch always maintains the best search performance. Moreover, the search time of \SG algorithm increases more rapidly than that of the other two algorithms. This is because the number of elements in each cell-based dataset increases as the $\theta$ increases and pairwise coverage computation and comparison will take more time. 

\myparagraph{Efficiency Comparison with Different $q$}
Fig.~\ref{fig:MCQCindexConstructed} shows the search time as $q$ increases. We can see that \coverageSearch consistently outperforms the other two algorithms. This is because our algorithm only needs to search once in each iteration by merging all results into one large query dataset to find connected datasets, reducing the search time. Moreover, we can see that \textsf{SG}+\unifiedIndex has better performance than \SG, which proves that our lower and upper bounds based on \unifiedIndex have stronger filtering power in finding the candidate datasets that satisfy spatial connectivity.


\myparagraph{Efficiency Comparison with Different $\delta$}
Fig.~\ref{fig:MCQCdelta} shows the search time of three algorithms as $\delta$ increases. We can see that the search time of the three algorithms gradually increases. This is because a higher $\delta$ means that more pairs of datasets satisfy the directly connected relation, resulting in more candidate datasets that satisfy the spatial connectivity, the algorithm needs to find the dataset with the maximum marginal gain by traversing all connected datasets. Moreover, we can see that \coverageSearch consistently outperforms the other two algorithms, which further shows the spatial merge strategy in \coverageSearch is efficient in solving \MCQC.


\subsubsection{Efficiency of Communication Cost}
Figs.~\ref{fig:MCQComcost} and \ref{fig:MCPTransmissionTime} present the communication cost and network transmission time of three algorithms. We can observe that as $q$ increases, the number of bytes transmitted by three algorithms gradually increases. The reason is that \coverageSearch only sends the query that overlaps with the MBR of the root node of the data source each iteration, without sending the entire query dataset.
Correspondingly, since the transfer time is proportional to the bytes transferred, the transfer time of \coverageSearch is more rapid than that of the other two algorithms, as shown in Fig.~\ref{fig:MCPTransmissionTime}.

\vspace{0.4cm}

 \section{Conclusions}
This paper mainly studied joinable search over multiple spatial data sources and formulated two search problems, including overlap joinable search problem (\MIQ) and coverage joinable search problem (\MCQC). To solve them, we first proposed a distributed tree-based spatial index structure \unifiedIndex, which is composed of one centralized global index and distributed local indices–one per data source. Next, based on \unifiedIndex, we designed query distribution strategies to reduce the communication cost. Subsequently, we proposed efficient lower and upper bounds and designed an exact \overlapSearch algorithm to solve \MIQ. Moreover, we also deduced the lower and upper bounds between two index nodes to accelerate the computation of spatial connectivity, and designed an approximate \coverageSearch algorithm to solve \MCQC efficiently. The experimental results showed that our proposed algorithms improved search performance and reduced communication costs. An interesting future research direction is to explore the spatial dataset search based on the data pricing to return the optimal dataset combination.
\newpage


\balance
\bibliographystyle{abbrv}
\bibliography{library}

\preto{\section}{\setcounter{cN}{0}} 


\clearpage
\balance
\section{Appendix}
\label{Appendix1}

\subsection{Notations}
\label{appendix:notations}
Frequently used notations are summarized in Table~\ref{tab:notations}.
\vspace{-0.5cm}
\begin{table}[h]
\centering
\setlength{\abovecaptionskip}{0cm} \caption{Summary of notations.}
 \label{tab:notations}
{ \begin{tabular}{cc}  
\toprule   
  \textbf{Notation} & \textbf{Description}  \\  
\midrule   
$\mathcal{D}$ & a data source consisting a set of spatial datasets\\
$\mathcal{S}_\mathcal{D}$ & a collection of cell-based datasets in $\mathcal{D}$\\
$D$ & a spatial dataset consisting of a set of points\\
$Q$, $S_Q$ & a query dataset, a cell-based query dataset\\
$D$, $S_D$ & a spatial dataset, a cell-based dataset\\
$n$ & the number of datasets in $\mathcal{D}$\\
$\theta$ & the resolution in the grid partition\\
$m$ & the number of data sources\\
$h, w$ & the height and width of the whole 2D space\\
$\mu, \nu$ & the height and width of each cell in the grid\\
$\mathcal{R}$ & the collection of result datasets \\
$2^\theta$ & the number of entries in each dimension\\
$N_Q$, $N_D$ & the query node and dataset node  \\
$N_{root}$, $N_{in}$ & the root node and internal node\\
$dist(S_{D},S_{D'})$ &the distance between $S_{D}$ and $S_{D'}$\\

  \bottomrule  
\end{tabular}}
\end{table}

\vspace{-0.2cm}
\subsection{NP-hardness of \MCQC}
\label{appendix:NPhard}
\begin{proof}
To prove \MCQC is NP-hard, we show that any \MCP instance can be reduced to an \MCQC instance in polynomial time. Considering a collection of sets $\mathcal{A}$ = $\{A_1, \dots, A_{|\mathcal{A}|}\}$ ($\mathcal{U}$ = $\cup_{A_i \in \mathcal{A}} A_i$) and a positive integer $k$, \MCP aims to find a subset $\mathcal{A}^* \subseteq \mathcal{A}$ such that $|\mathcal{A}^*| \leq k$ and $|\cup_{A_i \in \mathcal{A}^*} A_i|$ is maximized. We show that \MCP can be viewed as a special case of \MCQC. Given an arbitrary instance of \MCP, we sort the universe $\mathcal{U}$ according to the lexicographic order, then create a mapping from the sorted universe $\mathcal{U} = \{u_1, \dots, u_{|\mathcal{U}|}\}$ to an integer set $\mathcal{I} = \{0, 1, \dots, |\mathcal{U}-1|\}$. We can construct a $2^{\theta}\times 2^{\theta}$ grid in the space such that $2^{\theta}\times 2^{\theta} > |\mathcal{U}|$.

Then, we set the connectivity threshold $\delta = 2^{\theta}\sqrt{2}$, the query dataset $A_Q$ = $\{0,\dots, 2^{\theta}\times 2^{\theta}-1\}\backslash \mathcal{I}$. The set $\mathcal{A} \cup \{A_Q\}$ must satisfy the spatial connectivity. This is because $dist(A_i, A_j) (A_i, A_j \in \mathcal{A} \cup \{A_Q\})$ must not exceed $\delta$, ensuring the connectivity constraint is always satisfied. Moreover, it is evident that the total reduction is performed in polynomial time. Thus, the optimal solution for \MCQC must also be optimal for the corresponding \MCP. If \MCQC is not NP-hard, \MCP must not be NP-hard since any instance of \MCP can be converted to an instance of \MCQC. This contradicts the fact that \MCP is NP-hard \cite{Hochbaum1998}. Therefore, \MCQC is NP-hard.
\end{proof}


\vspace{-0.2cm}
\subsection{Index Update}
\label{appendix:indexUpdation}
As described in Section~\ref{sec:offline}, our local index is a bidirectional pointer structure, where each dataset node $N_D$ not only contains the pointer of child nodes $N_D.ch$, but also contains the pointer of parent node $N_D.pa$. Thus it can quickly support index updates without rebuilding the entire index. Since \globalIndex and \localIndex are similar structures, we take \localIndex as an example to illustrate the index update process, including dataset inserts, updates, and deletes.

Firstly, to insert a new node $N_D$, we need to add the dataset node to the appropriate leaf node. We first find the node $N_{in}$ with the minimum distance $||N_{in}.o, N_D.o||_2$ in each layer of the tree, and traversal continues down to the child nodes. When the tree node is a leaf node, we insert $N_D$ into the leaf node. If the leaf node capacity is greater than $f$, we need to split the node according to Algorithm~\ref{alg:indexConstructure}. Finally, we iteratively update the parent node information to reflect the properties of the new child nodes accurately.

Secondly, to update the existing dataset node $N_D$, we can find the original dataset node in the tree according to the unique dataset identifier $N_D.id$. Subsequently, we replace the original node $N_D$ with the updated dataset node $N_D'$ and iteratively update $(rect, o, r, N_l, N_r, pa)$ of the parent node from the bottom up until the parent node contains the child node completely. Finally, for a dataset node $N_D$ that needs to be deleted, we can consider it a special case of a dataset node update. We directly find and remove it from the leaf node and update the information of the parent node.

\subsection{Complexity Analysis}
\label{appendix:indexComplexity}

We use $n$ to represent the number of datasets in the data source and $m$ to represent the number of data sources.

\myparagraph{Time Complexity Analysis of Constructing \unifiedIndex}
Firstly, we analyze the time complexity of constructing \unifiedIndex, which comprises the construction of \globalIndex and \localIndex. 
In constructing \localIndex for each data source, Algorithm~\ref{alg:indexConstructure} splits the dataset nodes into two child sets from which trees are recursively built. Consequently, during each split process, it takes $O(n)$ time to split all dataset nodes, and the height of the tree is $\log{n}$. In addition, \localIndex needs to build the inverted index for each leaf node, resulting in a time complexity of $O(n|S_D|)$, where $|S_D|$ denotes the average size of the cell IDs contained by the dataset node $N_D$. Thus, the total time complexity of constructing \localIndex is $O(n\log{n}+n|S_D|)$.

In contrast, the construction process for \globalIndex is similar to that of \localIndex, resulting in a time complexity of $O(m\log{m})$. Since the number of data sources $m$ and the number of cell IDs in each dataset $|S_D|$ are significantly smaller than $n$, they can be considered negligible in comparison to $n$. Thus, the total time complexity of \unifiedIndex is $O(m\log{m}+ m(n\log{n}+n|S_D|))) = O(n\log{n})$.


\myparagraph{Space Complexity Analysis of Constructing \unifiedIndex}
Next, we analyse the time complexity of Algorithm~\ref{alg:indexConstructure} for constructing \unifiedIndex.
The space complexity primarily depends on the number of index nodes in \unifiedIndex. Firstly, in \localIndex of each data source, the height of the tree is $\log{n}$, resulting in a total number of index nodes given by the series $1+2+4+\dots+n = 2n$. Thus, the number of index nodes can be approximated as $O(n)$. Similarly, for \globalIndex with height $\log{m}$, the number of index nodes is $O(m)$. Assuming that the space required for each index node is $|size|$, and given that the number of data sources $m$ and $|size|$ are significantly smaller than $n$, they can be considered negligible in comparison. Thus, the total space complexity of \unifiedIndex is $O(mn|size| + m|size|) = O(n)$.

\myparagraph{Time Complexity Analysis of \overlapSearch} 
The time complexity of \overlapSearch mainly depends on the index structure and pruning strategy. In the worst-case scenario, the search may need to traverse all nodes of the tree. Thus, the time complexity of \overlapSearch is $O(n)$. However, due to the pruning strategy, the actual number of nodes accessed is usually much less than $n$. Thus, the time complexity is much lower than $O(n)$ in the actual situation.

\myparagraph{Time Complexity Analysis of \coverageSearch}
Here, we show the detailed time complexity analysis of \coverageSearch.
Algorithm~\ref{alg:MergeGreedy} shows that we iteratively choose the dataset with the maximum marginal gain and satisfy the connectivity. 
In the worst-case scenario, the algorithm may need to visit $n$ index nodes. Since the algorithm requires $k$ iterations,  the total time complexity is $O(kn)$. With the help of the acceleration strategy, the time complexity is much lower than $O(kn)$ in the actual situation.



\subsection{Approximation Guarantee of \coverageSearch}
\label{appendix:approximation}

In the following, we provide the approximation guarantee analysis of \coverageSearch, which achieves (1-1/e)-approximation under the condition that at the start of each iteration, the current result set can connect with at least one set in the optimal solution, and the marginal gain brought by the set is above the average value of all the gains brought by the optimal solution. The detailed  analysis is as follows.

First, we use $S_Q$ to denote the query set, $\mathcal{S}_\mathcal{D} = \{S_{D_1}, S_{D_2}, \dots, S_{D_n}\}$ to denote $n$ spatial datasets. Let $\mathcal{U} = (\cup_{i=1}^n S_{D_i})\cup S_Q$ be the set of all set elements, $\mathcal{R} = \{S_Q, S_{D_1}', \dots, S_{D_k}'\}$ denote the output of result set after $k$ iterations in our Algorithm~\ref{alg:MergeGreedy}, $C_i = (\cup_{j=1}^i S_{D_i}') \cup S_Q$ denote the set of elements covered by the end of iteration $i$, and $V = U - C_{i-1}$ denote the set of elements which are not in $C_{i-1}$. In addition, let the optimal solution be $\{S_Q, S_{D_1}^*, \dots, S_{D_k}^*\}$, and $CS_k = \cup_{j=1}^k S_{D_j}^*$.
Then we have the following Lemma~\ref{lemma:prooflabel1}. 
\begin{lemma1}
\label{lemma:prooflabel1} 
    At the start of each iteration $i$, if the $C_{i-1}$ is connected with at least one optimal set $S^* \in \{S_{D_1}^*, \dots, S_{D_k}^*\}$, where $|S^* \cap V| \geq \frac{1}{k} \sum_{j=1}^{k}|S_{D_j}^*\cap V|$, then we have
    \begin{small}
    \begin{equation}
        |S_{D_i}' \cap V| = |C_i| - |C_{i-1}| \geq \frac{|OPT| - |C_{i-1}|}{k}.
    \end{equation}
    \end{small}
\end{lemma1}
\begin{proof}
\renewcommand{\qedsymbol}{}
Let $OPT$ be the union of $CS_k$ and $S_Q$, i.e., $OPT = CS_k \cup S_Q$, and $|OPT|$ denote the number of elements in $OPT$. 
Thus, at the start of iteration $i$, the number of elements covered in $OPT$ but not in $C_{i-1}$ is 
\begin{small}
\begin{equation}
    |OPT - C_{i-1}| = |OPT| - |OPT \cap C_{i-1}|
    \geq  |OPT| - |C_{i-1}|.
\end{equation}
\end{small}
Moreover, $|(CS_k \cup S_Q) \cap V| = |OPT - C_{i-1}|$, that is
\begin{small}
\begin{equation}
\label{eq:setequ}
\begin{aligned}
  |(CS_k  \cup S_Q) \cap V|  &= |OPT \cap V| \\
  & = |OPT \cap (U-C_{i-1})| \\
  &= |(OPT \cap U) - (OPT \cap C_{i-1}) | \\
  &= |OPT| - |(OPT \cap C_{i-1})| \\
  &= |OPT - C_{i-1}|.
\end{aligned}
\end{equation}
 \end{small}
Therefore, we obtain that 
 \begin{small}
\begin{equation}
\begin{aligned}
  |(CS_k \cup S_Q) \cap V| \geq |OPT| - |C_{i-1}|.
\end{aligned}
\end{equation}
 \end{small}
In addition, we also have 
 \begin{small}
\begin{equation}
\label{eq:setDifference}
\begin{aligned}
  |(CS_k  \cup S_Q) \cap V| &= |(CS_k \cap V) \cup (S_Q \cap V)| \\
  &= |(CS_k \cap V)|.
\end{aligned}
\end{equation}
\end{small}
Then we derive that 
\begin{small}
\begin{equation}
\label{eq:set_OPT_Ci_1}
\begin{aligned}
  |(CS_k \cap V)| & \geq |OPT| - |C_{i-1}|.
\end{aligned}
\end{equation}
\end{small}
Dividing the two sides of Inequality~(\ref{eq:set_OPT_Ci_1}) by $k$, we obtain that 
 \begin{small}
\begin{equation}
\label{eq:set13}
\begin{aligned}
  \frac{|(CS_k \cap V)|}{k} &= \frac{|(\cup_{j=1}^k S_{D_j}^*) \cap V|}{k} \\
  & = \frac{|\cup_{j=1}^k (S_{D_j}^* \cap V)| }{k} \\
  &\geq \frac{|OPT| - |C_{i-1}|}{k}.
\end{aligned}
\end{equation}
\end{small}

We know that 
 \begin{small}
\begin{equation}
\label{eq:set14}
\begin{aligned}
  \sum_{j=1}^{k} |S_{D_j}^* \cap V| \geq |\cup_{j=1}^k (S_{D_j}^* \cap V)|. 
\end{aligned}
\end{equation}
\end{small}

By Inequalities~\ref{eq:set13} and \ref{eq:set14}, we obtain that
 \begin{small}
\begin{equation}
\label{eq:set15}
\begin{aligned}
  \frac{\sum_{j=1}^{k} |S_{D_j}^* \cap V|}{k} \geq \frac{ |\cup_{j=1}^k (S_{D_j}^* \cap V)| }{k} \geq \frac{|OPT| - |C_{i-1}|}{k}.
\end{aligned}
\end{equation}
\end{small}

We know that at least one optimal set $S^* \in \{S_{D_1}^*, \dots, S_{D_k}^*\}$ is connected with $C_{i-1}$, where
 \begin{small}
\begin{equation}
\label{eq:set16-1}
\begin{aligned}
    |S^* \cap V| \geq \frac{\sum_{j=1}^{k}|S_{D_j}^* \cap V|}{k}.
\end{aligned}
\end{equation}
\end{small}

Combining Iequalities~\ref{eq:set15} and \ref{eq:set16-1}, we obtain that
 \begin{small}
\begin{equation}
\label{eq:set17}
\begin{aligned}
  |S^* \cap V| \geq \frac{|OPT| - |C_{i-1}|}{k}.
\end{aligned}
\end{equation}
\end{small}

In addition, based on Lines~\ref{alg:findBest} to \ref{alg:MGCG-addBest} of Algorithm~\ref{alg:MergeGreedy}, we find the spatial dataset $S_{D_i}'$ with the maximum marginal gain in the iteration $i$. That is 
\begin{small}
\begin{equation}
\label{eq:set18}
    |S_{D_i}' \cap V| \geq |S^* \cap V|.
\end{equation}
\end{small}

Thus, by Inequalities~\ref{eq:set17} and \ref{eq:set18}, we conclude that: for all $ i = \{1, 2, \dots k\}$,
\begin{small}
\begin{equation}
\label{eq:set19}
|S_{D_i}' \cap V| = |C_i| - |C_{i-1}| \geq |S^* \cap V| \geq  \frac{|OPT| - |C_{i-1}|}{k}.
\end{equation}
\end{small}
\end{proof}
\vspace{-0.8cm}

\begin{lemma1}
\label{eq:setDifference2}
     For all $i = \{1, 2,\dots, k\}, |C_i| \geq \frac{|OPT|}{k}\sum_{j=0}^{i-1}(1 - 1/k)^j + (1-1/k)^i|S_Q|$.
\end{lemma1}
\begin{proof} 

We use induction reasoning to prove the correctness of Lemma~\ref{eq:setDifference2}. For convenience, we let $o=|OPT|/k$ and $t=(1-1/k)$.
At the start of our algorithm, the set of covered elements can be considered as $S_Q$, i.e. $C_0 = S_Q$.
Firstly, the base case $i=1$ is trivial as the first choice $|C_1|$ has at least $\frac{|OPT|}{k} + (1-1/k)|S_Q|$ elements by Lemma \ref{lemma:prooflabel1}.
Then we suppose inductive step $i$ holds, and prove that it holds for $i+1$. By Lemma~\ref{lemma:prooflabel1}, we obtain that
\begin{small}
\begin{equation}
\label{eq:set20}
\begin{aligned}
  |C_{i+1}| &\geq |C_i| + \frac{|OPT|-|C_i|}{k}\\
  &= o + t|C_i| \\
  &\geq o + t(o\sum_{j=0}^{i-1}t^j + t^i|S_Q|) \\
  &= ot^0 + o\sum_{j=1}^{i}t^j + t^{i+1}|S_Q| \\
  &= o\sum_{j=0}^{i}t^j + t^{i+1}|S_Q| \\
  &= \frac{|OPT|}{k}\sum_{j=0}^{i}(1-1/k)^j + (1-1/k)^{i+1}|S_Q|.
\end{aligned}
\end{equation}
\end{small}

The first inequality is by Lemma \ref{lemma:prooflabel1}, and the last inequality is from the inductive hypothesis.
\end{proof}

\noindent Finally, we have the following theorem and prove it.
\begin{theorem}
\label{theorem:coverage}
     At the start of each iteration $i$, if the $C_{i-1}$ is connected with at least one optimal set $S^* \in \{S_{D_1}^*, \dots, S_{D_k}^*\}$, where $|S^* \cap V| \geq \frac{1}{k} \sum_{j=1}^{k}|S_{D_j}^* \cap V|$, then Algorithm~\ref{alg:MergeGreedy} is a (1-1/e)-approximation algorithm for \MCQC.
\end{theorem}

\begin{proof}
Based on Lemma~\ref{eq:setDifference2}, we obtain that 
\begin{small}
\begin{equation}
\label{eq:set21}
\begin{aligned}
  |C_k| &\geq \frac{|OPT|}{k}\sum_{j=0}^{k-1}(1 - 1/k)^j + (1-1/k)^k|S_Q|\\
  &= \frac{|OPT|}{k}\times \frac{1-(1-1/k)^k}{1-(1-(1-1/k))} + (1-1/k)^k|S_Q|\\
  &= |OPT|(1-(1-1/k)^k) + (1-1/k)^k|S_Q| \\
  &\geq |OPT|(1-(1-1/k)^k).
\end{aligned}
\end{equation}
\end{small}

 Moreover, we know that $1+x \leq e^x$ for all $x \in \mathbb{R}$, which means $(1-1/k)^k \leq (e^{-1/k})^k = 1/e$.
 Then $1-(1-1/k)^k \geq 1-1/e$. Thus, we can derive that
\begin{small}
\begin{equation}
\begin{aligned}
|C_k| \geq (1-1/e)|OPT|,
\end{aligned}
\end{equation}
\end{small}
which proves that our greedy algorithm provides a (1-1/e)-approximation.
\end{proof}

\subsection{Data Source Details}
\label{appendix:dataset}
Here, we present the details of five data sources.
\begin{itemize}
[leftmargin=*]
\item[$\bullet$] \Baidu\footnotemark  is collected from the Baidu Maps Open Platform, which contains spatial datasets from different industry categories for 28 cities in China.
\item[$\bullet$] \BTAA\footnotemark is collected from the Big Ten Academic Alliance Geoportal, which contains spatial data for the midwestern US, such as Illinois, Indiana, and Michigan.
\item[$\bullet$] \NYU\footnotemark is collected from the NYU Spatial Data Repository, which contains geographic information on multi-subjects like census and transportation worldwide.
\item[$\bullet$] \Transit\footnotemark is collected from Big Geoportal, which contains different kinds of traffic data from Maryland and Washington D.C., such as buses, metro, and waterways.
\item[$\bullet$] \UMN\footnotemark is collected from the data repository of Minnesota University, which contains geographic information from various topics like agriculture and ecology.
\end{itemize}

\subsection{Additional Experimental Results}
\label{appendix:Experiment}
\myparagraph{Index Update Time Comparison}
we conduct experiments to investigate the performance of index updates, including batch inserts and updates of datasets. Specifically, we set the number of updated or inserted dataset nodes $\beta = 100, 150, 200, 250, 300$. We compare the update time of the five indexes with the increase in the number of node updates or insertions. We do not show experiments on the deletion of dataset nodes here, because the deletion of dataset nodes can be seen as a special case of the dataset node updates.

Fig.~\ref{fig:nodeInsert} shows the comparison of index updating time for five indexes as the dataset insertions increase. We observe that the index updating of \unifiedIndex is slower than \STS but faster than the others. This is because compared with \STS, \unifiedIndex needs to update the tree structure according to the dataset insertion strategies described in Appendix~\ref{appendix:indexUpdation}. However, since \unifiedIndex is a bidirectional pointer structure, it is faster in inserting than \quadtree and \Rtree. Additionally, we can also see that \Josie's index updating is consistently slower than the other indexes. This is primarily because \Josie takes more time in sorting the elements in each set and inserting the dataset into the posting list in sorted order.

Fig.~\ref{fig:nodeUpdate} compares the index updating time of five indexes as the number of dataset updates increases. We can observe that the index update time of \STS is still the fastest due to it can quickly update the dataset ID information in the corresponding posting list by traversing the updated cells. In contrast, \Josie is the second-fast. This is because the original dataset has already been sorted according to the dataset ID, and compared to \STS, \Josie also needs to update the sorted position information of the cell ID in the dataset. In addition, \unifiedIndex is ranked as the third, since it requires updating the tree structure in addition to the inverted index information. Similarly, \unifiedIndex has a faster index update time than \Rtree with the help of the bidirectional pointer structure.  In contrast, \quadtree has the slowest index update time because it has to repeatedly find the updated cell ID for insertion and deletion.

\begin{figure}[htbp]
\vspace{-0.4cm}
\setlength{\abovecaptionskip}{-0.3 cm} 
\setlength{\belowcaptionskip}{-0.7 cm} 
\centerline{\includegraphics[width=8cm]{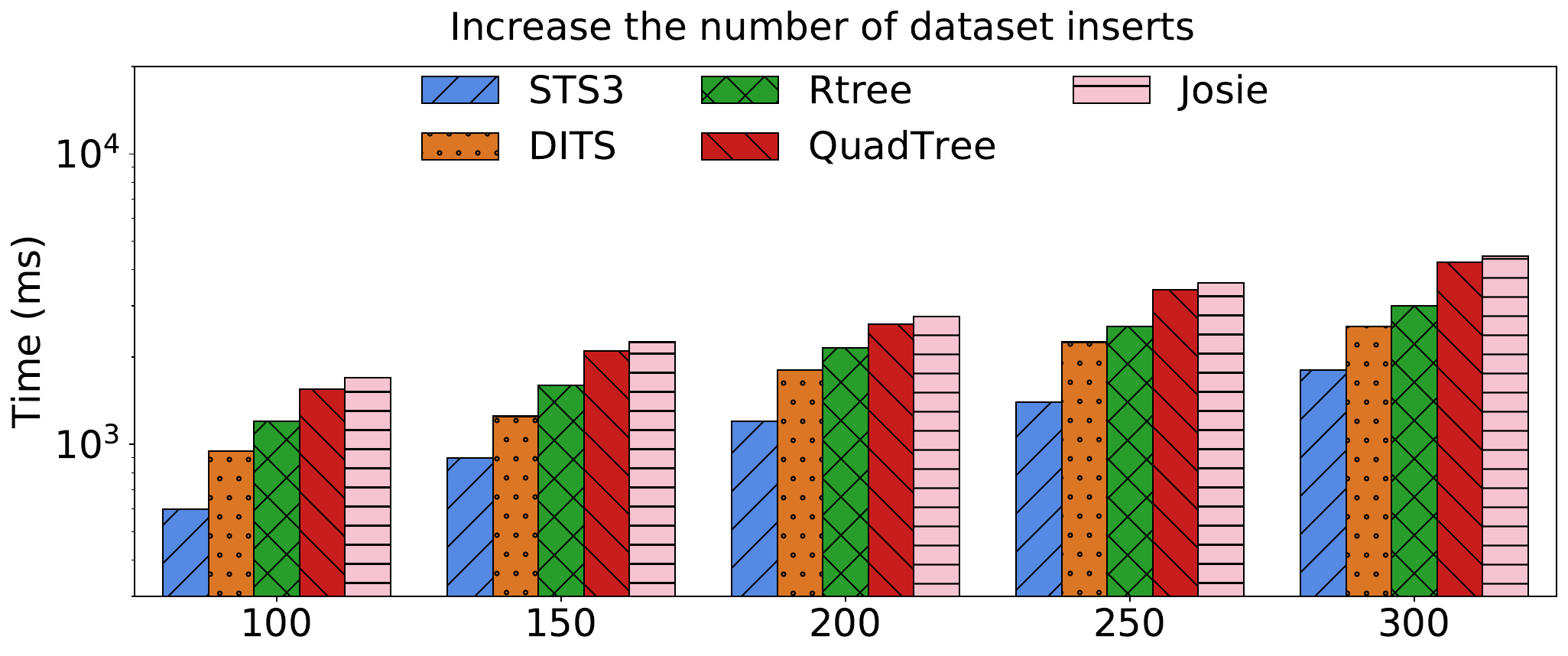}}
\caption{Index updating time comparison as the dataset insertions increase.}
\label{fig:nodeInsert}
\end{figure}

\begin{figure}[htbp]
\vspace{-0.6cm}
\setlength{\abovecaptionskip}{-0.3 cm} 
\setlength{\belowcaptionskip}{-0.7 cm} 
\centerline{\includegraphics[width=8cm]{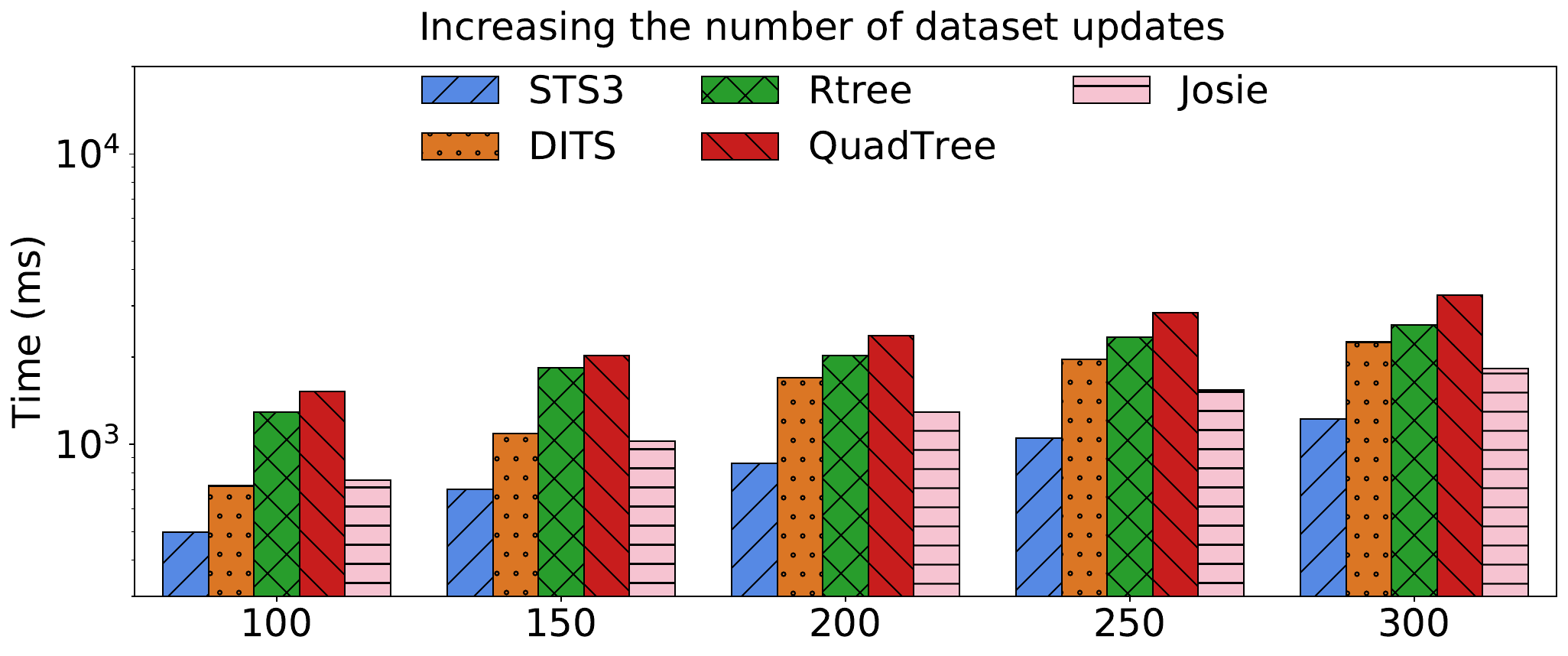}}
\caption{Index updating time comparison as the dataset updates increase.}
\label{fig:nodeUpdate}
\end{figure}

\footnotetext[2]{\url{https://lbsyun.baidu.com/}}
\footnotetext[3]{\url{https://geo.btaa.org/}}
\footnotetext[4]{\url{https://geo.nyu.edu/}}
\footnotetext[5]{\url{https://geo.btaa.org/}}
\footnotetext[6]{\url{ https://conservancy.umn.edu/drum}}


\end{document}